\documentclass[a4paper,11pt]{article}
\usepackage[numbers]{natbib}
\RequirePackage[colorlinks, linkcolor=blue,citecolor=blue,urlcolor=blue, linktocpage]{hyperref}
\usepackage[T1]{fontenc}
\usepackage{amsfonts}
\usepackage{verbatim}
\usepackage{amsmath, amsthm, amssymb, mathrsfs,graphicx,nicefrac,euscript,empheq,authblk, graphicx}
\usepackage{enumitem}
\usepackage{mathtools}
\usepackage[colorlinks]{hyperref}
\usepackage{graphicx}
\usepackage{subcaption}
\usepackage{}
\usepackage{epigraph,color}

\newtheorem{cor}{Corollary}

\theoremstyle{definition}
\newtheorem{remark}{Remark}

\renewcommand{\H}{\mathcal{H}}

\newcommand{\Be}{{\rm{Bern}}}

\newcommand{\eqd}{\stackrel{d}{=}}

\newcommand{\I}{{\rm{1}}}
\newcommand{\R}{{\mathbb{R}}}
\newcommand{\N}{{\mathbb{N}}}

\newcommand{\A}{\mathscr{A}}

\renewcommand{\ell}{L}

\renewcommand{\S}{S}

\renewcommand{\A}{\mathcal{A}}

\newcommand{\T}{\mathcal{T}}
\newcommand{\E}{{\mathrm E}}

\renewcommand{\i}{\mathrm{i}}

\newcommand{\C}{\mathbb{C}}

\def\R{\mathbb{R}}
\def\N{\mathbb{N}}

\newcommand{\eps}{\varepsilon}
\newcommand{\e}{e}

\newcommand{\argmin}{\operatornamewithlimits{arg\,min}}

\renewcommand{\Re}{\operatorname{Re}}
\renewcommand{\Im}{\operatorname{Im}}
\renewcommand{\P}{{\mathrm P}}

\newcommand{\B}{\mathcal{B}}

\renewcommand{\kappa}{\varkappa}

\newcommand{\F}{\mathcal{F}}

\newtheorem{thm}{Theorem}
\newtheorem{lem}{Lemma}
\newtheorem{defi}{Definition}

\theoremstyle{remark}
\newtheorem{ex}{Example}

\usepackage{pgfplots}
\usepgfplotslibrary{fillbetween}

\title{Statistical Inference for Quasi-Infinitely Divisible Distributions via Fourier Methods} 

\author[1]{Vladimir Panov \thanks{vpanov@hse.ru}}
\author[2]{Anton Ryabchenko \thanks{apryabchenko@hse.ru}}

\affil[1,2]{HSE University,
Pokrovsky boulevard 11, 109028 Moscow, Russia}

\begin{document}
\maketitle
\begin{abstract}
This study focuses on statistical inference for the class of quasi-infinitely divisible  (QID) distributions, which was  recently introduced by Lindner, Pan and Sato~\cite{LPS2018}.  The paper presents a Fourier approach, based on the analogue of the L{\'e}vy-Khintchine theorem with a signed spectral measure. In particular, this method allows to recover the components in a mixture model \(\mu=p  \tilde\mu_{\sigma} + (1-p) \mu^\circ\),  which is a QID distribution, where \(p>1/2\), \(\tilde\mu_\sigma\) is a centered normal distribution with unknown variance \(\sigma^2\), and the measure \(\mu^\circ\) satisfies a certain nonparametric condition. We prove that for some subclasses of QID distributions, the considered estimates have polynomial rates of convergence. This is an interesting finding when compared to the logarithmic convergence rates of similar methods for infinitely divisible distributions, which cannot be improved in general. We demonstrate the numerical performance of the algorithm using simulated examples. \end{abstract}

\section{Introduction} 
A random variable \(X\) has quasi-infinitely divisible (QID) distribution if there exist two random variables \(Y\) and \(Z\) with infinitely divisible (ID) distributions such that \begin{eqnarray}\label{def1}
X+Y \eqd Z,
\end{eqnarray} where \(X\) and \(Y\) are independent. Trivially, this class includes all ID distributions, but it is actually significantly larger. In particular, it includes several notable examples of distributions, which are not ID, such as the Bernoulli distribution with parameter that is not equal to 1/2 and some mixtures of normal distributions with different variances.

A complete characterization of quasi-infinite divisibility is not known yet, but there are some simple criteria for certain subclasses of distributions. For instance,  a distribution on integers  is QID if and only if its characteristic function doesn't have real zeros (Lindner, Pan and Sato~\cite{LPS2018}). More generally,  a discrete distribution is QID if and only if its characteristic function is separated from zero (Alexeev and Khartov~\cite{AK2023}). This result also holds for  mixtures of discrete and absolutely continuous distributions (Berger and Kutlu~\cite{BK2023}). Along with the mentioned papers on the QID random variables, there are several papers dealing with the generalisation of this concept to random vectors (Berger et al.~\cite{BKL2022}) and stochastic processes (Passegeri~\cite{Pass2020}). Note that these notions are not only of theoretical interest, they  have  found applications in financial modelling (Madan et al. \cite{Madan2023}),  physics, specifically in the analysis of Landau levels (Chhaiba et al.~\cite{Ch2016}, Demni and Mouayn \cite{DM2015}), number theory (Nakamura~\cite{Nakamura}) and insurance mathematics (Zhang et al.~\cite{Zhang2014}). 

It is worth mentioning that  this research area has emerged relatively recently, starting with the pioneer work~\cite{LPS2018}. Prior to that, QID had only been mentioned in some contexts without a systematic study  (Cuppens~\cite{Cuppens}, Linnik and Ostrovsky~\cite{LO1977}). 

Our research deals with the statistical inference for QID distributions. In contrast to the extensive literature on statistical estimation for ID distributions and L{\'e}vy-based models, the estimation for a QID distribution has been considered in only one paper by Passegeri~\cite{Pass2023} in the framework of Bayesian analysis. In the current paper, we consider a Fourier-based approach, which relies on the analogue of the well-known L{\'e}vy\,--\,Khintchine formula for the QID distributions. This analogue yields that any QID distribution can be described by a triplet \((\gamma, \sigma^2, \nu),\) where the first two elements \(\gamma \in \R,\) \(\sigma \in \R_+\) are the same as in the ID case, while the third, \(\nu\), is a signed spectral measure. This triplet is unique for any QID distribution. \newline

\textit{Contribution.} 
In this paper we consider two related problems: 
\begin{enumerate}
\item estimation of the  triplet \((\gamma, \sigma^2, \nu)\) of a QID distribution using a Fourier-based approach;
\item recover the components of a mixture model \(\mu = p \mu_1 + (1-p)\mu_2\), where \(\mu_1\) is known up to a parameter, and the unknown distribution \(\mu_2\) satisfies a condition guaranteeing that \(\mu\) is a QID distribution.
\end{enumerate}
The first problem was studied in the ID case, in particular, for L{\'e}vy processes (Belomestny~\cite{DB2010}, Comte and Genon-Catalot~\cite{CGC2010},
 Gugushvili~\cite{Gugu}, Neumann and Rei{\ss}~\cite{NR2009}), but, to the best of our knowledge, it is  new in the context of QID distributions, which are not ID.
 
As for the second problem,  our particular interest will be to the distributions, which are the mixtures of the normal distribution \(\tilde\mu_{\sigma}\) with zero mean and variance equal to \(\sigma^2\), and an absolutely continuous distribution \(\mu^\circ\), that is, 
\begin{eqnarray}\label{model1}
\mu &=&p  \tilde\mu_{\sigma} + (1-p) \mu^\circ, \qquad p \in (0,1).
\end{eqnarray}
This distribution is QID, if \(p>1/2\) and \(\mu^\circ\) is the distribution with  characteristic function \(\phi^\circ\) such that 
\begin{eqnarray}
\label{H}
\mathcal{H}(u):=\phi^\circ(u) e^{u^2 \sigma^2 /2}, \qquad u \in \R,
\end{eqnarray} is the characteristic function of a probability measure.

%
This setup leads to a type of contamination model, where the clean data from a normal distribution $\tilde\mu_{\sigma}$ are contaminated by a second mixture component, which can be represented, e.g., by any QID distribution with the second element of the L{\'e}vy triplet larger than \(\sigma^2\) (see Example~\ref{ex3} below). The term ``contamination'' was first introduced by Huber \cite{Huber} and has since been applied in various fields, starting from the pioneering works by Mandelbrot~\cite{Mandelbrot} in finance, and Brownie, Habicht and Robson \cite{Brownie} in chemistry.

In our paper, we provide a semiparametric procedure for the estimation of unknown parameters \(p\) and \(\sigma^2\), and unknown distribution \(\mu^\circ\), from the observations of \(\mu.\) 
One widely used tool for solving estimation problems for mixtures is the EM algorithm which is a parametric method that requires assumptions on the form of the mixture components, most commonly used for normal mixtures (see McLachlan and Krishnan~\cite{McLahclanEM}, McLachlan and Peel~\cite{McLachlan}). However, semi-parametric EM algorithms also exist, but they have their own limitations. For example, it is often assumed that both distributions, \(\mu_1, \mu_2\), or only the second component \(\mu_2\) are symmetric around some unknown value, resulting in the location-shifted mixtures (Bordes, Chauveau, Vandekerkhove \cite{Bordes}, Hohmann and Holzmann~\cite{Hohman}, Bordes and Vandekerkhove~\cite{Bordes2010}, Hunter,  Wang  and Hettmansperger~\cite{Hunter}). We would like to emphasise that our paper makes use of completely different assumptions, which lead to the identification of both components.

The parametric part of our algorithm (estimation of \(p\) and \(\sigma^2\)) is based on the representation of the characteristic function in a form similar to the L{\'e}vy\,--\,Khintchine formula.  We show that the estimates  have polynomial rates of convergence if the distribution of \(\mu^\circ\) is supersmooth. This fact is  an interesting observation when compared to the fact that the convergence rates of similar methods for ID distributions are logarithmic and can't be improved in general, see Theorems~4.3, 4.4, 5.7 in Belomestny and Reiss~\cite{BR2015}. To estimate \(\mu^\circ\), we propose the kernel estimate, which also has a polynomial rate of convergence, given our choice of the estimates of \(p\) and \(\sigma^2\). Our approach provides a  comprehensive semi-parametric estimation procedure for the model~\eqref{model1}, which, to the best of our knowledge, does not have any analogues in the literature for the considered generality. \newline

\textit{Structure.} The paper is organised as follows. In the next section, we provide a brief overview of quasi-infinitely divisible distributions.  Next, in Section~\ref{sec3} we propose a Fourier-based approach for all elements of the characteristic triplet of a QID distribution.  Later, in Section~\ref{sec4}  we describe a semiparametric estimation procedure for the model~\eqref{model1}, which relies on the ideas presented in Section~\ref{sec3} and the kernel estimates for the density function of the measure \(\mu^\circ\).
 Our main theoretical findings are given in Section~\ref{sec5}. Comparison of Theorems~\ref{thm1} and \ref{thm2} yields that the rates  of convergence for the model~\eqref{model1} are essentially faster than in the general case if the distribution with characteristic function \(\mathcal{H}\) is supersmooth. 
Section~\ref{secnum} deals with the empirical study and includes, in particular, numerical comparison with the EM algorithm for the normal mixtures. The proofs are collected in Section~\ref{sec7}.

 \section{Quasi-infinitely divisible distributions}
 \subsection{Definition and spectral representation}\label{sec21}
The  L{\'e}vy\,--\,Khintchine formula states that the characteristic function \(\phi\) of an ID distribution can be represented in the following form \begin{eqnarray} \label{LC}
    \phi(u)
    =
    \exp\Bigl\{
      \i \gamma u-\frac{1}{2}\sigma^2 u^2
			    +
		    \int_{\mathbb{R}\setminus \{0\}}
		    \left( 
		    		e^{\i u x}-1 - \i u c(x) 		\right)			\nu(dx)\Bigr\},
\end{eqnarray}
where \(\gamma \in \R, \sigma \in \R_+,\) \(\nu: \B(\R \setminus \{0\}) \to \R_+\) is a measure such that  \(\int_{\R}  \min(1, x^2)\nu(dx) < \infty,\)  and 
$c: \mathbb{R} \rightarrow \mathbb{R}$ is a so-called representation function, which is a bounded, Borel measurable function, satisfying \(c(0)=0\) and
\begin{equation}\label{repr_function}
c(x) = x+ o(x^2) \qquad \text{as} \quad x \to 0.
\end{equation}
The function $c(x)$ is typically chosen as $c(x) = x\I_{[-1, 1]}(x)$.

To introduce the quasi-infinitely divisible (QID) distributions and an analogue of~\eqref{LC} for the QID case, it is convenient to merge \(\sigma\) and \(\nu\) into a single finite measure  \(\zeta\) on \(\B(\R):\)
\begin{eqnarray*}
\zeta(dx) &:=& \sigma^2 \delta_0(dx) + \min(1,x^2) \I_ {\{ x\ne 0\}}\nu(dx), 
\end{eqnarray*}
which allows to rewrite~\eqref{LC} in the following form  
\begin{eqnarray} \label{LC2}
\phi(u) &=& \exp\Bigl\{ 
\i \gamma u+ \int_{\R} g(x,u) \zeta(dx)
\Bigr\}, 
\end{eqnarray}
where
\begin{eqnarray*}
g(x,u) &:=& \begin{cases} 
\frac{e^{\i u x}-1 - \i u c(x) }{\min(1, x^2)}, & x \ne 0,\\
-u^2 /2,  & x = 0.
\end{cases}  \end{eqnarray*}
Trivially, there is a one-to-one correspondence between the  pair $(\gamma, \zeta)$ and the initial L{\'e}vy triplet \((\gamma, \sigma^2, \nu)\), provided that the function $c(x)$ is fixed. In fact,
\begin{equation}
\label{zeta}
    \sigma^2 = \zeta(\{0\}) \quad \text{and} \quad \nu(dx) =  \bigl(\min(1, x^2)\bigr)^{-1}\zeta(dx), \qquad  x \in \R \setminus \{0\}.
\end{equation}
The main idea behind the concept of QID, is to replace the measure \(\zeta\) in~\eqref{LC2} by a signed measure.  Recall that 
a \textit{signed measure} $\zeta$ on $\mathbb{R}$ is a function $\zeta: \mathcal{B} (\R) \rightarrow [-\infty, \infty]$ on the Borel $\sigma$-algebra $\mathcal{B}$ such that
\(
    \zeta(\emptyset) = 0 \) and \(\zeta(\bigcup_{j=1}^{\infty}A_j) = \sum_{j=1}^{\infty} \zeta(A_j)
\)
for all sequences $(A_j)_{j \in \mathbb{N}}$ of pairwise disjoint sets in $\mathcal{B}(\R)$, where the   value of the series does not depend on the order of the $A_j$, i.e., the series converges unconditionally.  
A signed measure $\zeta$ is finite, if $\zeta(A) \in \mathbb{R}$ for all $A \in \mathcal{B}(\R)$. 

Now we are in a position to define QID distributions.
\begin{defi} 
A distribution \(\mu\) is QID, if its characteristic function admits the representation \eqref{LC2} with some $\gamma \in \mathbb{R}$ and a finite signed measure $\zeta: \B(\R) \to [-\infty, +\infty]$.
\end{defi}
Similarly to the ID distributions, the characteristic pair $(\zeta, \gamma)$ uniquely determines the QID distribution for a fixed representation function $c(x)$. 
Note that not every pair $(\gamma, \zeta)$ forms a QID distribution. Otherwise, let $(\gamma,\zeta)$ be a characteristic pair of a QID distribution $\mu$ and $(n^{-1}\gamma, n^{-1}\zeta)$ be a characteristic pair of a QID distribution $\mu_n$. Then $\mu$ must be infinitely divisible since $\mu = \mu^{*n}_n$ for any $n \in \N$, but this leads to $\zeta$ being a positive measure, which is a contradiction.

For the study of QID distributions, it is convenient to apply the Hahn\,--\,Jordan decomposition, which states that for a finite signed measure $\zeta$, there exist disjoint Borel sets $C_{+}$ and $C_{-}$ and unique finite measures $\zeta_{+}$ and $\zeta_{-}$ on $\B(\R)$ such that 
\begin{equation*}
    \zeta_{+}(\mathbb{R} \setminus C_{+}) = \zeta_{-}(\mathbb{R} \setminus C_{-}) = 0 \quad \text{and} \quad \zeta = \zeta_{+} - \zeta_{-}.
\end{equation*}
 Combining this with~\eqref{LC2} leads to the following representation 
\begin{eqnarray}\label{phi12}
\phi(u) = \frac{\phi_1(u)}{\phi_2(u)},
\end{eqnarray}
where \(\phi_1\) and \(\phi_2\) are characteristic functions of ID distributions with  characteristic pairs \((\gamma,\zeta_+)\) and \((0,\zeta_-). \) The equivalence with the definition given at the beginning of the introduction is now clear. 

Note that for any QID distribution 
\begin{eqnarray*}
 \sigma^2:= \zeta(\{0\}) = -2 \lim_{u\to \infty} \bigl( \log(\phi(u)) / u^2 \bigr) \geq 0,
\end{eqnarray*}
see Lemma~2.7 in~\cite{LPS2018}, and therefore the characteristic pair \((0, \zeta_{-})\) is equivalent to the L{\'e}vy triplet in the form \((0,0, \nu_{-}).\) Therefore, \eqref{phi12} yields  the following analogue of the L{\'e}vy\,--\,Khintchine formula~\eqref{LC} for QID distributions
\begin{multline} \label{LC3}
    \phi(u)
    =
    \exp\Bigl\{
      \i \gamma u-\frac{1}{2}\sigma^2 u^2
			    +
		    \int_{\mathbb{R}\setminus \{0\}}
		    \left( 
		    		e^{\i u x}-1 - \i u c(x)			\right)			\nu_+(dx)
		\\ -
	    \int_{\mathbb{R}\setminus \{0\}}
		    \left( 
		    		e^{\i u x}-1 - \i u c(x)			\right)			\nu_-(dx)				\Bigr\},
\end{multline}
where \(\nu_+\) and \(\nu_{-}\) are determined by \(\zeta_+\) and \(\zeta_{-}\) as in~\eqref{zeta}.
 Therefore, the characteristic triplet \((\gamma, \sigma^2, \nu)\) with \(\nu=\nu_+ - \nu_{-}\) uniquely determines the QID distribution.
\subsection{Examples}\label{sec22} 

Many interesting examples can be constructed using the following lemma, which is proven in~\cite{LPS2018}, Corollary 3.5.

\begin{lem} \label{magic_lem1}
Consider the mixture $\mu = p\mu_1 + (1-p)\mu_2$, where  $p \in (1/2, 1)$,  $\mu_1$ is a QID distribution with  triplet $(\gamma, \sigma^2, \nu)$, and $\mu_2$ is some distribution on \(\R.\) Suppose  that there exists a finite signed measure $\Lambda$ on $\mathbb{R}$ such that its Fourier transform is equal to 
\begin{eqnarray}\label{mu12}
\F\bigl[\Lambda \bigr](u) = \frac{\phi_2(u)}{\phi_1(u)}, \qquad u \in \R,
\end{eqnarray}
where   $\phi_1$ and $\phi_2$ are the characteristic functions of $\mu_1$ and $\mu_2$, If this measure is finite with total variation $|\Lambda|(\mathbb{R}) < p/(1-p)$, then $\mu$ is QID with the characteristic triplet \[\Bigl(\gamma + \int_{\mathbb{R}}c(x)\Tilde{\nu}(\mathrm{d}x),  \quad \sigma^2, \quad  \nu + \Tilde{\nu}\Bigr),\]
where $\Tilde{\nu}$ is a finite signed measure defined as follows: 
\begin{equation}\label{Tildenu}
    \Tilde{\nu}(B) = \sum_{m = 1}^{\infty} \frac{(-1)^{m + 1}}{m}\Bigl( 
\frac{1-p}{p} 
\Bigr)^m\Lambda^{*m}(B), \qquad B \in \B(\R\setminus\{0\}),
\end{equation}
with $\Lambda^{*m}(B)$ being the convolution of the measure $\Lambda$ with itself  $m$ times.\end{lem}
The proof of Lemma~\ref{magic_lem1} is based on the representation of the characteristic function of \(\mu\) in the form 
\begin{eqnarray*}
\F\bigl[
\mu 
\bigr](u) 
=
p \phi_1(u)
\exp\Bigl\{
\log
\Bigl(
1+ 
\frac{1-p}{p} 
\F\bigl[
\Lambda
\bigr](u) \Bigr)
\Bigr\},
\end{eqnarray*}
and further an application of the Taylor series of the function \(\log(1+z).\)

%
%


Note that the assumption $|\Lambda|(\mathbb{R}) < p/(1-p)$ is fulfilled if \(\Lambda\) is a probability measure, since \(p \in (1/2,1)\). The next two corollaries present specific cases when \(\Lambda\) is a probability measure, and therefore Lemma~\ref{magic_lem1} can be applied.\begin{cor}
(Cuppens theorem.) Any distribution \(\mu\) with an atom of mass larger than 1/2 is QID. 
\end{cor}
\begin{proof}
In fact, in this case \(\mu = p \mu_1 + (1-p) \mu_2,\) where \(p=\mu(\{\lambda\})>1/2\), \(\mu_1 = \delta_\lambda\), and 
\(\mu_2 = (\mu - p \delta_\lambda)/(1-p).\) Note that \(\mu_2\) is a probability measure, and \(\Lambda=\mu_2 * \delta_{-\lambda}\)   is also a probability measure. We conclude that   \(\mu\)  is a QID distribution with triplet \((\lambda,0, \nu),\)  where 
\begin{eqnarray*}
\nu = \sum_{m=1}^\infty \frac{(-1)^{m+1}}{m}
\Bigl( 
\frac{1-p}{p} 
\Bigr)^m \bigl(  \mu_2 * \delta_{-\lambda} \bigr)^{*m}.
\end{eqnarray*}
\end{proof}
\begin{ex}\label{ex1}
In particular, the Bernoulli distribution \(\Be(p)\) with probability of success \(p \ne 1/2\) is QID. Recall that this distribution is not ID, since its support is bounded (see Proposition~8.1 from~\cite{SvH2004}), except for the trivial cases  $p \in \{0, 1\}$.  Note that \(\Be(1/2)\)  is also not QID, since its characteristic function has real zeros, 
\(\phi(u) = (1+  e^{i u})/2 = 0 \) 
 for \(u = (2k+1) \pi,\) what is not possible due to representation~\eqref{phi12} and the fact that the characteristic function \(\phi_1(\cdot)\)  of an ID distribution doesn't have zeros (see Proposition~2.8 from~\cite{SvH2004}).
 \end{ex}
\begin{cor} \label{cor2}
Consider the mixture $\mu = p\tilde\mu_\sigma + (1-p) \mu^\circ$, where  \(\tilde\mu_\sigma\) is a normal distribution  with parameters \(0\) and \(\sigma^2\), and \(\mu^\circ\) is the distribution with the characteristic function \(\phi^\circ\) such that the function \(\mathcal{H}(u)\) presented in~\eqref{H}
is the characteristic function of some probability measure \(\Lambda\). Then  for any \(p \in (1/2,1]\), \(\mu\) is a QID distribution  with the characteristic triplet \[\Bigl( \int_{\mathbb{R}}c(x)\Tilde{\nu}(\mathrm{d}x),  \quad  \sigma^2, \quad \Tilde{\nu}\Bigr),\]
where  $\Tilde{\nu}$ is given by~\eqref{Tildenu}.
\end{cor}
\begin{proof}
For \(p \in (1/2, 1),\) the corollary directly follows from Lemma~\ref{magic_lem1}, since the measure \(\Lambda\) is a probability measure. For \(p=1\) the statement follows from the fact that \(\tilde{\mu}_\sigma\) is an ID distribution with triplet \((0, \sigma^2,0)\).
\end{proof}
\begin{ex}\label{ex2}
In particular, a mixture of two normal distributions $$\mu = p\tilde\mu_{\sigma_1} + (1-p) \tilde\mu_{\sigma_2}$$  is QID if the main component has smaller variance, that is, \(p>1/2\) and \(\sigma_1 < \sigma_2\) or, vice versa, \(p<1/2\) and \(\sigma_1 > \sigma_2\). For instance, in the first case, the function \(\mathcal{H}(u) = e^{- u^2(\sigma_2^2 -\sigma_1^2)} , u \in \R,\) is a characteristic function of the normal distribution with zero mean and variance equal to \(\sigma_2^2 -\sigma_1^2.\) A mixture of two normal distributions with different variances is an example of a QID distribution, which is not ID, because its tail is too thin (see Chapter~VI from~\cite{SvH2004}). 
\end{ex}
\begin{ex} \label{ex3}
More generally, the assumptions of Corollary~\ref{cor2} are fulfilled if \(\mu^\circ\) is QID with characteristic triplet \((\gamma, \bar\sigma^2, \nu)\) if \(\bar{\sigma} > \sigma.\) In fact, due to the L{\'e}vy\,--\,Khintchine formula~\eqref{LC3}, \(\mathcal{H}(u)=\phi^\circ(u) e^{u^2 \sigma^2 /2}\) is again a characteristic function of a QID distribution with characteristic triplet \((\gamma, \bar\sigma^2 - \sigma^2, \nu).\) 
\end{ex}
\section{Fourier-based estimation procedure}
\label{sec3}
%
In this section we extend the Fourier-based estimation procedure for the ID distributions described in~\cite{BR2015} to the QID distributions.

For simplicity we consider  the case when the jump parts of  \(Y\) and \(Z\) in \eqref{def1} have absolutely continuous distribution of compound Poisson type. In terms of the  measures $\nu_{+}$ and $\nu_{-}$, this means that for any \(B \in \B(\mathbb{R}),\)
\begin{equation*}
    \nu_{+}(B) = \lambda_{+}\int_{B}p_{+}(x)\mathrm{d}x, \qquad \nu_{-}(B) = \lambda_{-}\int_{B}p_{-}(x)\mathrm{d}x,
\end{equation*}
where \(\lambda_{\pm} = \nu_{\pm}(\R) \geq 0\), and 
$p_{\pm}$  are some probability  density functions.  Note that the quasi-L{\'e}vy measure \(\nu\) has a density \(s(x) = \lambda_+ p_+(x) - \lambda_- p_-(x)\), \(x\in \R.\)


In our setup, formula~\eqref{LC3} can be simplified to 
\begin{eqnarray}\nonumber     \phi(u)
    &=&
    \exp\Bigl\{
      \i \gamma^* u-\frac{1}{2}\sigma^2 u^2
			    +
		    \int_{\mathbb{R}}
		    \left( 
		    		e^{\i u x}-1	\right)			\nu_+(dx) -
	    \int_{\mathbb{R}}
		    \left( 
		    		e^{\i u x}-1 	\right)			\nu_-(dx)				\Bigr\}\\		\label{phiexp}
&=&  \exp\Bigl\{
      \i \gamma^* u-\frac{1}{2}\sigma^2 u^2
			    +
\F\bigl[s \bigr](u) -  \lambda^*\Bigr\},
\end{eqnarray}
where 
\begin{eqnarray*}
    \gamma^* = \gamma - \int_{\R} c(x)s(x) dx, \qquad 
    \lambda^* = \nu(\R) = \lambda_{+} - \lambda_{-},
\end{eqnarray*}
and
$\F[s](u) = \int_\R e^{\i u x} s(x) dx$ is the Fourier transform of \(s.\) Since \[\int_\R s(x) dx = \lambda^* <\infty,\]  the  Riemann-Lebesgue lemma yields
\(
  \bigl| \F[s](u) \bigr| \rightarrow 0\) as \(u \to \infty\).
Therefore, as \(u \to \infty\), 
\begin{eqnarray}\label{e1}
        \Re \bigl( \log(\phi(u)) \bigr) &=&  -\frac{1}{2} \sigma^2 u^2 - \lambda^{*} +o(1), \\
        \label{e2}
        \Im \bigl( \log(\phi(u)) \bigr) &=&\gamma^{*}u + o(1).
\end{eqnarray}

These ideas give rise for the following 4-step estimation procedure for the estimation of the parameters \(\gamma^*,\sigma, \lambda^* \) and the density \(s\) from the observations \(X_1,..,X_n\) of the corresponding QID distribution. 
\begin{enumerate}
  \item First, the characteristic function of \(X\) is estimated by \begin{equation}\label{em_char_func}
      \phi_n(u) = \frac{1}{n}\sum_{k=1}^{n} e^{iuX_k}.
  \end{equation}
  \item Then, motivated by~\eqref{e1},  introduce the estimate  
  \begin{eqnarray}\nonumber
\bigl( \sigma_n^2, \lambda^*_n\bigr) &=& \argmin_{\sigma^2, \lambda^* } \int_{\R_+} w^{U_n}(u) 
\Bigl[
\Re\bigl(\log(\phi_n(u)) \bigr) 
+\frac{1}{2} \sigma^2 u^2  +\lambda^* 
\Bigr]^2 du,\\
&&\label{sigmalambda}
\end{eqnarray}
where \(w^{U_n}(u)=U_n^{-1} w(u/U_n)\) with some unbounded increasing sequence of positive numbers \(U_n\), and a function \(w\) is supported on \([\eps,1]\), \(\eps>0,\) and belongs to the class \(L^1([\eps,1] )\). Direct calculations yield that the solution $\sigma^2_n$ of~\eqref{sigmalambda} can be represented in the form
\begin{equation}\label{sigma_sol}
    \sigma^2_n = \int_{\R_+}w^{U_n}_{\sigma^2}(u) \Re\bigl(\log(\phi_n(u)) \bigr)du,
\end{equation}
where
\begin{equation*}
    w^{U_n}_{\sigma^2}(u) = \frac{2 w^{U_n}(u) \Bigl( u^2\int_{\R_+} w^{U_n}(s) ds - \int_{\R_+} w^{U_n}(s) s^2 ds\Bigr)}{\bigl(\int_{\R_+} w^{U_n}(s) s^2 ds\bigr)^2 - \int_{\R_+} w^{U_n}(s) s^4ds \cdot \int_{\R_+} w^{U_n}(s) ds}.
\end{equation*}
Note that $w^{U_n}_{\sigma^2}(u)=U_n^{-3}w^{1}_{\sigma^2}(u/U_n)$, and moreover,
\begin{equation}\label{prop1}
    \int_{\R_+} w^{U_n}_{\sigma^2}(u) du = 0, \qquad \int_{\R_+} \frac{-u^2}{2} w^{U_n}_{\sigma^2}(u) du = 1.
\end{equation}
Similarly,
\begin{equation}\label{lambda_sol}
    \lambda^*_n = \int_{\R_+}w^{U_n}_{\lambda^*}(u) \Re\bigl(\log(\phi_n(u)) \bigr)du,
\end{equation}
where the function $w^{U_n}_{\lambda^*} (u) = U_n^{-1}w^{1}_{\lambda^*}(u/U_n)$ has the following properties
\begin{equation*}
    \int_{\R_+} -w^{U_n}_{\lambda^*}(u) du = 1, \qquad \int_{\R_+} \frac{u^2}{2} w^{U_n}_{\lambda^*}(u) du = 0.
\end{equation*}
 \item Analogously, motivated by~\eqref{e2}, introduce the estimate
 \begin{eqnarray}\label{gamma_opt}
\gamma^*_n
= \argmin_{\gamma^*} \int_{\R_+} w^{V_n}(u) 
\Bigl[
\Im\bigl(\log(\phi_n(u)) \bigr) 
-\gamma^* u
\Bigr]^2 du,
\end{eqnarray}
where \(V_n\) is some unbounded increasing sequence of positive numbers, and the function \(w\)  can be defined in the same way as in the previous step.
The solution of this problem is given by 
\begin{equation}\label{gamma_sol}
    \gamma^*_n = \int_{\R_+}w^{V_n}_{\gamma^*}(u) \Im\bigl(\log(\phi_n(u)) \bigr)du,
\end{equation}
where the function $w^{U_n}_{\gamma^*}(u)=U_n^{-2}w^{1}_{\gamma^*}(u/U_n)$ satisfies
\begin{equation*}
    \int_{\R_+} u w^{U_n}_{\gamma^*}(u) du = 1.
\end{equation*}
   
  \item Finally, we turn towards nonparametric estimation of the density \(s\). From~\eqref{phiexp}, we get that the natural estimate of \(s\) is the inverse Fourier transform, that is,
\begin{eqnarray}
\nonumber
s_n(x) 
&=& \F^{-1}\Bigl[\bigl(
\log(\phi_n(\cdot)) - \i \gamma^*_n(\cdot) +\frac{1}{2} \sigma_n^2 (\cdot)^2
 + \lambda^*_n \bigr)w_s(\cdot / T_n)\Bigr] (x), \label{p}
\end{eqnarray}
where $x \in \R$ and $w_s$ is a weight function supported on $[-1, 1]$, and  \(T_n\) is some unbounded increasing sequence of positive numbers.
\end{enumerate}

\section{Semiparametric inference for the model~\eqref{model1}}
\label{sec4}
As we have seen in Corollary~\ref{cor2}, the measure
\begin{eqnarray}\label{model2}
\mu &=&p  \tilde\mu_{\sigma} + (1-p) \mu^\circ.
\end{eqnarray}
 is QID, if \(p>1/2\) and the function \(\mathcal{H}\) defined by~\eqref{H} is a characteristic function of some distribution.  In this section, we propose a semiparametric estimation procedure for the estimation of the known parameters \(p\) and \(\sigma^2\) and unknown measure \(\mu^\circ \), which is assumed to be absolutely continuous. 
 \newline\newline
\!\!\!\textit{Parametric part.} Let us point out two important facts related to the model~\eqref{model2}.
\begin{enumerate}
\item  In this particular case, the measure \(\Lambda\), which plays an essential role in Lemma~\ref{magic_lem1}, is a probability measure. Due to this fact,  formula~\eqref{Tildenu} yields
\begin{eqnarray*}
\lambda^*=
\Tilde\nu (\R) = \sum_{m=1}^\infty \frac{(-1)^{m+1}}{m}
\Bigl( 
\frac{1-p}{p} 
\Bigr)^m \underbrace{\Lambda^{*m}(\R)}_{=1} = -\log(p), 
\end{eqnarray*}
where the last equality follows from \((1-p)/p \in (0,1).\)
\item The second element of the characteristic triplet, \(\sigma^2\), coincides with the variance of the first (normal) component in~\eqref{model2}.
\end{enumerate}
Therefore, the application of the first and second steps from the previous section results into the estimation of \(p = e^{-\lambda^*}\) by the natural plug-in estimate \(p_n := e^{-\lambda_n^*}, \) and \(\sigma^2\) by \(\sigma_n^2.\)
\begin{remark} \label{rem2a}
From these two facts, it follows that both parameters \(p\), \(\sigma^2\) and the measure \(\mu^\circ\) can be uniquely recovered from the measure \(\mu\). Indeed, let \((p_1, \sigma_1^2, \mu^\circ_1)\) and \((p_2, \sigma_2^2, \mu^\circ_2)\) be two different vectors of parameters corresponding to the measure $\mu$ from (\ref{model2}). 
    As we have seen in Section~\ref{sec21}, the characteristic triplet corresponding to the QID distribution $\mu$ is unique. The exact form of the characteristic triplet for the mixture model is given in Lemma~\ref{magic_lem1}. From this form it follows that $\sigma^2_1 = \sigma^2_2$. Parameter $p$ is equal to \(e^{-\nu(\R)}\) and therefore depends only on the quasi-L{\'e}vy measure of $\mu$, which is uniquely determined by $\mu$. We conclude that  $p_1 = p_2$. Finally,  \(\mu_j^\circ\), $j = 1, 2$ can be recovered as \(\mu_j^\circ = \bigl( \mu - p_j  \tilde\mu_{\sigma_j} \bigr)/(1-p_j)\), which leads to $\mu_1^\circ = \mu_2^\circ$. Therefore the vectors \((p_1, \sigma_1^2, \mu^\circ_1)\) and \((p_2, \sigma_2^2, \mu^\circ_2)\) are identical.



\end{remark}
\textit{Nonparametric part.} Now we turn towards the non-parametric estimation of the measure \(\mu^\circ\) in~\eqref{model2}. Assuming that this measure (or equivalently \(\mu\)) is absolutely continuous, we first estimate the density \(g\) of  \(\mu\) by the kernel estimate 
\begin{eqnarray*}
g_n(x) =  \frac{1}{nh}\sum_{i = 1}^{n} K\Bigl(\frac{X_i - x}{h}\Bigr), \qquad x \in \R,
\end{eqnarray*}
where $K: \R \to \R_{+}$ is a kernel function satisfying \(\int_{\R} K(u) du = 1\). Then, we apply linear transform to estimate the density \(g^\circ\) of \(\mu^\circ\),
\begin{eqnarray}\label{g_n_estim}
g^{\circ}_{n} (x) := \frac{g_n (x)- p_n \varphi_{\sigma_n}(x) }{1- p_n},
\end{eqnarray}
where \(\varphi_{\sigma_n}\) is the density of the normal distribution with zero mean and variance equal to \(\sigma_n^2.\) However, this estimate has the undesired property of producing negative values. To address this, we propose to use the positive part estimate, which is defined as
\begin{equation}\label{pos_part_g_n}
    g^{\circ +}_{n} (x) = \max(g^{\circ}_{n} (x), 0).
\end{equation}
Note that the estimator \(g^{\circ +}_{n} (x)\) might not be a density function, but the estimators of this type are widely used in the statistical literature, see, e.g., \cite{Tsybakov}.

\section{Convergence analysis}\label{sec5}
\subsection{General case of QID distributions}
\label{sec6}
We begin this section with a  general result, which shows that, similar to the ID case, the rates of the procedure described in Section~\ref{sec3} are logarithmic.
Consider the subclass $\mathcal{S}(r, \bar{\sigma}, C)$ of QID distributions with characteristic functions of the form~\eqref{phiexp}, where the density \(s\) satisfies \(\int_{|x|>1 }|x| s(x)dx < C\), is \(r\)-times differentiable and moreover
\begin{eqnarray}\label{condcond}
\sigma \in (0, \bar\sigma],  \quad |\gamma^*| \leq C, \quad \lambda^* \in [0, C], \quad \max\bigl\{
\| s^{(r)}\|_2, \| s^{(r)}\|_\infty\bigr\} \leq C.
\end{eqnarray}
with some $\bar\sigma, C >0$.
The conditions~\eqref{condcond} are in line with previous research on the statistical inference of ID distributions using Fourier-based methods, see, e.g.,  \cite{BR2015} and \cite{Gugu}.




 It is convenient to establish the rates of convergence conditional to some event \(\A_n\), which we define in the following lemma. The first part of this lemma shows that the probability of this event tends to 1 at a polynomial rate, if a constant \(\chi^\circ\) is small enough. The second part states that the sequence \(\chi_n\), which is essential in the definition of \(\A_n\), converges to zero under certain choice of the sequence \(U_n.\)
\begin{lem}\label{lem1}
Denote the event 
\begin{equation*}
    \mathcal{A}_n = \mathcal{A}_n (\chi^\circ) := \Bigl\{\max_{u \in [-U_n,U_n]} \frac{|\phi_n(u) - \phi(u)|}{|\phi(u)|} \leq \chi_n \Bigr\}, \qquad n=1,2,....
\end{equation*}
with \[\chi_n := \chi^\circ \frac{\sqrt{\log(nU_n^2)/ n}}{\inf_{u \in [-U_n, U_n]}|\phi(u)|},\]
where \(\chi^\circ\) is positive constant.
The following statements hold.
\begin{enumerate}
\item[(i)] If \(\chi^\circ < 1/16,\) then 
 \[\P\{\A_n\} \geq 1-c (\sqrt{n}U_n)^{-\kappa}\] with \(\kappa= (1/(2\chi^\circ)^2-64) / 128>0\) and some positive constant \(c,\) which  depends on \(\E[|X_1|]\) only.
 \item[(ii)] For any distribution from the class \(\mathcal{S}(r, \bar{\sigma}, C)\),  it holds 
 \begin{eqnarray*}
 \chi_n\lesssim  e^{C}
  \sqrt{\frac{\log(nU_n^2)}{n}}
\exp(\frac{1}{2}\overline{\sigma}^2U_n^2),
\end{eqnarray*}
where  $ \lesssim $ stands for inequality up to an absolute constant not depending on \(n\) and the parameters of the class. In particular, if \(U_n \leq \sqrt{\log n /\sigma^2_{max}}\) with any $\sigma_{max} > \overline{\sigma}$, we have  \(\chi_n \lesssim e^C n^{-\lambda}\) with \(\lambda<(1-(\bar{\sigma}/\sigma_{max})^2)/2\).

 \end{enumerate}\end{lem}
\begin{proof}
(i) Using a version of Hoeffding's inequality, given in Proposition~3.3 in~\cite{BR2015}, we get 
     \begin{eqnarray*}
         \P\bigl\{\mathcal{A}_n^c \bigr\}&\leq& \P\Bigl\{\max_{u \in [-U_n,U_n]}|\phi_n(u) - \phi(u)| > \chi_n \inf_{u \in [-U_n, U_n]}|\phi(u)|\Bigr\} \\
    &=& \P\Bigl\{ \sqrt{n/\log(nU_n^2)}\max_{u \in [-U_n,U_n]} |\phi_n(u) - \phi(u)| > \chi^\circ \Bigr\}
    \leq c (\sqrt{n}U_n)^{-\kappa},
     \end{eqnarray*}
provided \(\chi^\circ<1/16\).

\item[(ii)] The lower bound for $\inf_{u \in [-U_n, U_n]}|\phi(u)|$ follows from the representation~\eqref{phiexp}
\begin{align*}
    \inf_{u \in [-U_n, U_n]}|\phi(u)| &=  \inf_{u \in [-U_n, U_n]}\Bigl|
    \exp\bigl\{
      \i \gamma^* u-\frac{1}{2}\sigma^2 u^2			    +
\F\bigl[s \bigr](u) -  \lambda^*\bigr\}
\Bigr| \\
    &= \inf_{u \in [-U_n, U_n]}\Bigl[\exp\bigl\{-\lambda^*-\frac{1}{2}\sigma^2u^2 + \Re(\F[s](u))\bigr\}\Bigr]\\
    &\gtrsim e^{-C}\exp(-\frac{1}{2}\overline{\sigma}^2U_n^2),
\end{align*}
where we use that $|\Re(\F[s](u))| \leq 1$.
\end{proof}

Note that the condition \(\int_{|x|>1 }|x| s(x)dx < C\)  is needed to ensure that $\E |X|$ (and therefore the constant \(c\) in (i)) is bounded uniformly on the class $\mathcal{S}(r, \bar{\sigma}, C)$. This condition can be relaxed to \(\int_{|x|>1 }|x| s_+(x)dx < C\), where \(s_+\) is the density of \(\nu_+,\) but the determination of \(s_+\) is not necessary for the rest of the paper.

\begin{thm}\label{thm1}

Let the weight function \(w \in L^1([\eps,1] )\) with  \(\eps>0\), and moreover, the functions \(w^{1}_{\sigma^2}, w^{1}_{\lambda}, w^{1}_{\gamma}\) defined by \eqref{sigma_sol}, \eqref{lambda_sol}, \eqref{gamma_sol},
belong to the class 
\begin{equation*}
\mathcal{P}_r := \Bigl\{ 
f: \;\; \F \bigl[ 
\frac{f(x)}{x^r}\bigr](\cdot)\in L^1 \bigl(\R \bigr)
\Bigr\},
\end{equation*}
the function \(w_s \in L^2([-1,1])\) is \(r\)-H{\"o}lder continuous and satisfies \(w_s(0)=1\). Let the sequence \(U_n=V_n\) and the constant \(\chi_\circ\) be such that \(\P\{\A_n\} \to 1\) and \(\chi_n \to 0\) as \(n \to \infty,\)  see Lemma~\ref{lem1}. 
Then on the event $\A_n$,
\begin{eqnarray*}
\sup_{\mathcal{S}}
\bigl|\sigma^2_n - \sigma^2
\bigr| \lesssim \frac{\mathcal{R}_n}{U_n^2},  \qquad
\sup_{\mathcal{S}}
\bigl|\lambda_n^* - \lambda^*
\bigr| \lesssim\mathcal{R}_n, \qquad
\sup_{\mathcal{S}}
\bigl|\gamma^*_n - \gamma^*
\bigr| \lesssim \frac{\mathcal{R}_n}{U_n},
\end{eqnarray*}
and
\begin{eqnarray*}
\sup_{\mathcal{S}} \int_{\R} \bigl( 
s_n(x) - s(x) 
\bigr)^2 dx \lesssim 
\chi_n^2 T_n + \mathcal{R}_n^2  T_n \Bigl( 1+  \frac{T_n}{U_n^2}  + \frac{T_n^2}{U_n^4}\Bigr)
 + C T_n^{-2r},
\end{eqnarray*}
where
\begin{eqnarray*}
\mathcal{R}_n =
\mathcal{R}_n (C, \overline{\sigma},r) := 
\e^{C}
\frac{\sqrt{\log(nU_n^2)}}{\sqrt{n}}
e^{\frac{1}{2}\overline{\sigma}^2U_n^2} +   \frac{C}{U_n^{r + 1}},
\end{eqnarray*}
and the supremums are taken over all distributions in the considered class \(\mathcal{S}=\mathcal{S}(r, \bar{\sigma}, C).\) 
 \end{thm}
%
\begin{remark} \label{rem1}
As we will see from the proof, the first summands in the upper bounds for the estimates \(\sigma_n^2, \lambda_n^*, \gamma_n^*\)  appear due to the terms of the form
\begin{eqnarray*}
\int_{\R} a_n(u)\Re\bigl(\log(\phi_n(u)) - \log(\phi(u))\bigr)du, 
\end{eqnarray*} 
where \(a_n(\cdot)\) is either \(w^{U_n}_{\sigma^2}(\cdot), w^{U_n}_{\lambda}(\cdot)\) or \(w^{U_n}_{\gamma}(\cdot)\). These terms are associated with the statistical errors of the estimates. The second summands appear due to 
\begin{eqnarray*}
\int_{\R}a_{n}(u)\F \bigl[s\bigr](u)du,
\end{eqnarray*}
which is closely related to the misspecification error (bias). The choice $$U_n =V_n= \sqrt{\frac{\log n}{\sigma^2_{max}}}, \qquad T_n = \bigl( \log n \bigr)^{\beta},$$ where $\sigma_{max} > \overline{\sigma}$ and \(\beta=\frac{r+1}{2(2r+1)}\), yields the logarithmic convergence rates on $\A_n$, 
\begin{align*}
\sup_{\mathcal{S}(r, \bar{\sigma}, C)}
\bigl|\sigma^2_n - \sigma^2
\bigr| \lesssim \;C
\Bigl(\frac{\log n}{\sigma^2_{max}}\Bigr)^{\frac{-(r + 3)}{2}}, \\[1em]
\sup_{\mathcal{S}(r, \bar{\sigma}, C)}
\bigl|\lambda^*_n - \lambda^*
\bigr| \lesssim \; C \Bigl(\frac{\log n}{\sigma^2_{max}}\Bigr)^{\frac{-(r + 1)}{2}}, \\[1em]
\sup_{\mathcal{S}(r, \bar{\sigma}, C)}
\bigl|\gamma^*_n - \gamma^*
\bigr| \lesssim \; C \Bigl(\frac{\log n}{\sigma^2_{max}}\Bigr)^{\frac{-(r + 2)}{2}},\\
\sup_{\S} \int_{\R} \bigl( 
s_n(x) - s(x) 
\bigr)^2 dx \lesssim 
C \frac{\bigl(\log n\bigr)^{\frac{-r(r+1)}{2r+1}}}{\min\bigl\{\sigma_{max}^{-(r+1)}, 1\bigr\}}.
\end{align*}
As we have seen  in Lemma~\ref{lem1}(i), this choice of \(U_n=V_n\) also guarantees that \(\P\{\A_n\}\) converges to 1 at a polynomial rate.
 \end{remark}
 \begin{remark} Let us briefly show what happens in the case of infinite quasi-L{\'e}vy measure \(\nu.\)  The general idea is to study our approach under model misspecification, that is, to use the estimates designed for finite  quasi-L{\'e}vy measures \(\nu,\)  for the distributions with infinite \(\nu\). For instance, for the estimate \(\sigma_n^2\) we have 
 \begin{eqnarray*}
  \sigma^2_n - \sigma^2 &=& \int_{\R} w^{U_n}_{\sigma^2}(u)\Re\bigl(\log(\phi_n(u)) - \log(\phi(u))\bigr)du \\
&&    \hspace{2cm}+ \Bigl( 
\int_{\R} w^{U_n}_{\sigma^2}(u)\Re\bigl(\log(\phi(u))\bigr)du 
- \sigma^2\Bigr),
\end{eqnarray*}
where, similar to Remark~\ref{rem1}, the first summand is associated with the statistical error, and can be treated in exactly the same way as in the proof of Theorem~\ref{thm1}. The second term is bias, and due to~\eqref{sigma_sol} and \eqref{prop1} it can be represented as  
\begin{align*}
\int_{\R} w^{U_n}_{\sigma^2}(u)\Re\bigl(\log(\phi(u))\bigr)du 
&- \sigma^2\\
&=
\int_{\R} w^{U_n}_{\sigma^2}(u)
\Bigl[\Re\bigl(\log(\phi(u))\bigr)+ \frac{1}{2} \sigma^2 u^2
\Bigr]du \\
&=
\int_{\R} w^{U_n}_{\sigma^2}(u)
\Bigl[\int_{\R} \bigl(
\cos(ux) - 1 
\bigr) \nu(dx) \Bigr]du\\
&=
U_n^{-2}
\int_{\R} w^{1}_{\sigma^2}(v)
\Bigl[\int_{\R} \bigl(
\cos(U_n v x) - 1 
\bigr) \nu(dx) \Bigr]dv.
\end{align*}
 Further consideration of the internal integral depends on assumptions on the quasi-L{\'e}vy measure \(\nu\). For example, following Section~6.1 from~\cite{BR2015}, let us assume a stable-like behaviour of the jump component, that is, 
 \begin{eqnarray*}
\int_{\R} \bigl(
1- \cos(ux) 
\bigr) \nu(dx) = c u^{a_1} +O(u^{a_2}), \qquad u \to \infty,
\end{eqnarray*}
with \(0 \leq a_2 <a_1<2\) and \(c>0\). In this particular case, we get \[\int_{\R} w^{U_n}_{\sigma^2}(u)\Re\bigl(\log(\phi(u))\bigr)du 
- \sigma^2 \lesssim U_n^{a_1-2}.\] This leads to the conclusion that on the event \(\A_n,\)
\begin{align*}
\sup_{\mathcal{S}(r, \bar{\sigma}, C)}
\bigl|\sigma_n^2 - \sigma^2
\bigr| &\lesssim \;
\e^{C}
\frac{\sqrt{\log(nU_n^2)}}{\sqrt{n}U_n^{2}}
e^{\frac{1}{2}\overline{\sigma}^2U_n^2} +   \frac{C}{U_n^{2-a_1}},
\end{align*}
resulting in logarithmic rates of convergence for \(\sigma_n^2\), see Remark~\ref{rem1}.
 \end{remark}

\subsection{Mixture models in the form~\eqref{model1}}
Now we turn towards statistical inference for a subclass of QID distributions in the form~\eqref{model1}, as well as the analysis of the estimates $\sigma^2_n$ and \(p_n\) introduced in Section~\ref{sec4}. The next theorem aims to show that  the convergence rates are polynomial, if the distribution of \(\mu^\circ\) is supersmooth.  

\begin{thm}\label{thm2}
Assume that the assumptions of  Theorem~\ref{thm1} are fulfilled. Consider the subclass \(\mathcal{S}^*(r, \bar{\sigma}, C, \underline{p}) \subset \mathcal{S}(r, \bar{\sigma}, C)\)  of QID distributions in the form~\eqref{model1} such that 
 \(p\geq \underline{p}>1/2\), and the function \(\mathcal{H}\) defined by~\eqref{H} is a characteristic function of some probability measure \(\Lambda\).
 Moreover assume that the distribution \(\Lambda\) is supersmooth in the sense that its characteristic function \(\mathcal{H}\) has exponential tails,
\begin{equation}\label{ss}
\bigl|\mathcal{H}(u) \bigr|  \leq c_1 e^{-c_2 |u|^\gamma}, \qquad u \in \R,
\end{equation}
with some \(c_1, c_2, \gamma >0.\)
Then on the event $\A_n$,
\begin{align*}
\sup_{\mathcal{S}^*}
\bigl|\sigma^2_n - \sigma^2
\bigr| \lesssim \frac{\breve{\mathcal{R}}_n}{U_n^2},  \qquad
\sup_{\mathcal{S}^*}
\bigl|\lambda_n^* - \lambda^*
\bigr| \lesssim\breve{\mathcal{R}}_n, \qquad
\sup_{\mathcal{S}^*}
\bigl|\gamma^*_n - \gamma^*
\bigr| \lesssim \frac{\breve{\mathcal{R}}_n}{U_n},
\end{align*}
and
\begin{eqnarray*}
\sup_{\S} \int_{\R} \bigl( 
s_n(x) - s(x) 
\bigr)^2 dx \lesssim 
\chi_n^2 T_n + \breve{\mathcal{R}}_n^2  T_n \Bigl( 1+  \frac{T_n}{U_n^2}  + \frac{T_n^2}{U_n^4}\Bigr)
 + C T_n^{-2r},
\end{eqnarray*}
where 
\begin{eqnarray*}
\breve{\mathcal{R}}_n =
\breve{\mathcal{R}}_n (C, \tilde{c}_1, c_2, \overline{\sigma}, \gamma) := 
\e^{C}
\frac{\sqrt{\log(nU_n^2)}}{\sqrt{n}}
e^{\frac{1}{2}\overline{\sigma}^2U_n^2} + \tilde{c}_1   \e^{-c_2 \eps^{\gamma} U_n^{\gamma}},
\end{eqnarray*}
and 
\(\tilde{c}_1 = - c_1  \log\bigl(2 \underline{p}-1 \bigr)>0\).
\end{thm}
\begin{remark}\label{remark2}
    Let us again, as in Lemma~\ref{lem1}(ii) and Remark~\ref{rem1}, choose
 $$U_n = V_n = \sqrt{\frac{\log n}{\sigma^2_{max}}}, 
$$ where $\sigma_{max} > \overline{\sigma}$. Then 
    \begin{eqnarray*}
\breve{\mathcal{R}}_n \lesssim C_1 \sqrt{\log(n)} n^{-C_2},
\end{eqnarray*}
where 
  \begin{eqnarray}
\label{CC1}
C_1 = \max\{e^{1+C}, \tilde{c}_1\}, \qquad C_2 = \min \Bigl(\frac{\sigma^2_{max} - \overline{\sigma}^2}{2\sigma^2_{max}}, \frac{ c_2\eps^{\gamma}}{\sigma^{\gamma}_{max}} \Bigr) > 0.\end{eqnarray}
Therefore, under this choice of the sequences \(U_n\)    and \(V_n\), the rates of convergence presented in Theorem~\ref{thm2} are polynomial. In particular, on the event \(\A_n\), 
\begin{eqnarray*}
\sup_{\mathcal{S}^*}
\bigl|\sigma^2_n - \sigma^2
\bigr| \lesssim  C_1\sigma_{max}^2  n^{-C_2}.
\end{eqnarray*}
and, moreover, 
\begin{align*}
\sup_{\mathcal{S}^*}
\bigl|p_n - p
\bigr| \lesssim C_1 \sqrt{\log(n)} n^{-C_2}. 
\end{align*}
If we choose 
\begin{eqnarray*}
T_n = n^{\beta}, \qquad \mbox{with} \quad \beta=2C_2/(2r+3),
\end{eqnarray*}
we get that the rate of convergence of \(s_n\) is also polynomial, namely, 
\begin{eqnarray*}
\sup_{\mathcal{S}^*} \int_{\R} \bigl( 
s_n(x) - s(x) 
\bigr)^2 dx \lesssim 
n^{-\min(2r\beta, 2\lambda-\beta)}.
\end{eqnarray*}
Note that \(2\lambda-\beta>0\) due to Lemma~\ref{lem1}(ii) and~\eqref{CC1}.
\end{remark}

\begin{remark} In this remark, we would like to  comment on the 
 assumption~\eqref{ss}. 
Note that~\eqref{ss} holds  with \(\gamma=2\)  for our key examples~\ref{ex2} and \ref{ex3}.  
It is also fulfilled for any distribution in the form~\eqref{model1} such that \(\mu^\circ\) is supersmooth with \(\gamma\geq2.\) Moreover, it is also fulfilled if \(\mu_0\) is a convolution of a supersmooth distribution with any \(\gamma\) (also \( \gamma \in (0,2)\)) and a normal distribution with variance \(\bar{\sigma}^2>\sigma^2\). Note that in the latter case  \(\mu_0\) is also supersmooth with the same \(\gamma.\)
\end{remark}
\begin{remark}The estimates obtained using similar methods  for ID distributions do not converge at a polynomial rate, even when we restrict ourselves to the supersmooth class. As an example, Theorem~2.5 from~\cite{Gugu} shows that for some subset \(\T\) of supersmooth ID distributions, 
\begin{eqnarray}\label{gugu}
\sup_{\T} \E\Bigl[
\bigl(\sigma_n^2 - \sigma^2 \bigr) ^2
\Bigr] \lesssim
\bigl( 
\log n
\bigr)^{-(\gamma+5)/2} \exp\Bigl\{
-2 c_2 ( \log n)^{\gamma/2}
\Bigr\},
\end{eqnarray}
leading to the conclusion that the rate is faster than any power of \(\log n.\) Note that the case \(\gamma\in(0,2)\) is the most important (see p.291 in~\cite{Gugu} or p.25 in  \cite{BT2008I}), and in this case~\eqref{gugu}  does not imply that the rate is polynomial.\end{remark}

The last result of this section shows that the nonparametric estimate of the second component of the mixture~\eqref{model2} has also polynomial rate of convergence, provided that the parameters are estimated with polynomial accuracy.
\begin{thm}\label{thm3}
Assume that the assumptions of Theorem~\ref{thm2} are fulfilled, and the parameter \(U_n\), constants $C_1$ and \(C_2\) are fixed as in  Remark~\ref{remark2}. Consider the subclass \(\mathcal{S}^{**}=\mathcal{S}^{**}(r, \underline{\sigma}, \overline{\sigma}, C, \underline{p}, \overline{p}) \subset \mathcal{S}^*(r, \bar{\sigma}, C, \underline{p}) \) consisting of the distributions with $\sigma \in \bigl[\underline{\sigma}, \overline{\sigma}\bigr]$  and $p \in \bigl[ \underline{p}, \overline{p}\,\bigr]$ for some \(0<\underline{\sigma}<\overline{\sigma}\) and $1/2 < \underline{p} < \overline{p} < 1$. Moreover, assume that $g^\circ$ is twice differentiable and  $g^\circ, (g^\circ)'' \in L^2(\R)$. Let \(K\) be a kernel of order 1 with
\begin{eqnarray*}
\int_\R K^2(u)du < \infty, \qquad \int_\R u^2 K(u) du <\infty.
\end{eqnarray*}
Then the conditional mean integrated square error w.r.t. $\A_n$ for $g_n^{\circ+}$ is given by
\begin{multline}\label{stat}
\sup_{\mathcal{S}^{**}}  \int_{\R} \E \Bigl[ \bigl(g^{\circ+}_{n} (x) - g^{\circ}(x)\bigr)^2 | \A_n\Bigr] dx  \\ 
    \lesssim \frac{1}{(1-\overline{p})^4}\Bigl(\frac{1}{nh}\int_{\R} K^2(u)du + \frac{h^4}{4}\bigl(\int_{\R}|u^2K(u)|du\bigr)^2 
    \int_\R \bigl( g''(x) \bigr)^2 dx
    \\
    + C_3 \log(n) n^{-2 C_2}\Bigr),
\end{multline}
where $C_3 =  C^2_1\bigl( 
\underline{\sigma}^{-1} + \int_{\R} g^2(x) \, dx
\bigr).$
\end{thm}

\begin{remark}\label{rem6}
Since the bandwidth parameter $h$ does not appear in the last term of the right-hand side of~\eqref{stat}, the method for selecting its optimal value is the same as that used for the standard kernel density estimation problem. Balancing the rates of the first two terms we get that the optimal bandwidth is equal to  $h = C n^{-1/5}$, where \(C\) is any positive constant, which is similar to the result given in Theorem~1.2 in~\cite{Tsybakov}. This selection ensures that the first two terms in the upper bound of Theorem~\ref{thm3} are of order $\mathcal{O}(n^{-4/5})$, resulting in a polynomial convergence rate for the estimate $g^{\circ+}_{n}$.
\end{remark}
%

\section{Numerical study}\label{secnum}
In this section we apply our algorithm to several models of the form~\eqref{model1}. First, we study the classic variance mixture model of two normal distributions. Then we proceed to the Bart Simpson distribution (known also as ``the claw''), which is a mixture of $6$ normal distributions. For both models we compare our approach with the EM algorithm, which is a widely used iterative method for estimating mixture parameters (see McLachlan, Peel \cite{McLachlan}), based on the maximisation of the log-likelihood. For both examples the iterative EM algorithm was terminated when either the maximum number of iterations ($K_{\text{max}} = 200$) was reached or the iteration step improvement in log-likelihood fell below $10^{-6}$.  Finally we analyze the performance of our approach for Example~\ref{ex3}, which involves the convolution of a normal distribution and another ID distribution as a second component of the mixture. We omit the comparison with the EM algorithm for the third model in order to emphasize that the proposed approach is semi-parametric and can estimate the second mixture component without making any initial assumptions about its specific form.

For all the examples below,  we choose the bandwidth parameter equal to $h = \frac{1}{30} n^{-1/5}$, following the discussion in Remark \ref{rem6}. The kernel function $K(x)$ is chosen to be the Epanechnikov kernel, i.e. $K(x) = \frac{3}{4}(1 - u^2), \,\,|u| \leq 1$, which is known to be optimal in a sense of minimizing the mean integrated squared error for the densities that belong to the Sobolev class of order $2$, see Section~1.2.4 from~\cite{Tsybakov}.

\subsection{Two-component normal mixture}
The first considered model is a mixture of two normal distributions, where the resulting density has the following form
\begin{equation*}
    g(x) = p\varphi_{\sigma_1}(x) + (1 - p)\varphi_{\sigma_2}(x), \qquad x \in \R,
\end{equation*}
with $0 < \sigma_1 < \sigma_2$, $1/2 < p < 1$.
The resulting distribution is QID as it was discussed in Example \ref{ex2}. For the numerical example, we fix the parameters of this distribution as $p = 0.75, \, \sigma_1^2 = 0.1, \, \sigma^2_2 = 0.5$.

Figure~\ref{fig1} shows the plot of real parts of logarithms of  $\phi(u)$ (orange line) and several realisations of its estimate $\phi_n(u)$ given by (\ref{em_char_func})  for \(n=1000\). As it can be seen in the graph, the deviation increases with the growth of $u$. This observation influences the choice of $U_n=V_n$, because on the interval \([\eps U_n, U_n]\), one should simultaneously have \(\phi_n(u) \approx \phi(u)\)  and \(\phi(u) \approx \i \gamma^*(u) - (1/2)\sigma^2 u^2 - \lambda^*.\) For instance, for \(n=1000\) these conditions are fulfilled with  \(U_n=V_n=8.\)

\begin{figure}[ht]
    \centering
    \includegraphics[scale=0.25]{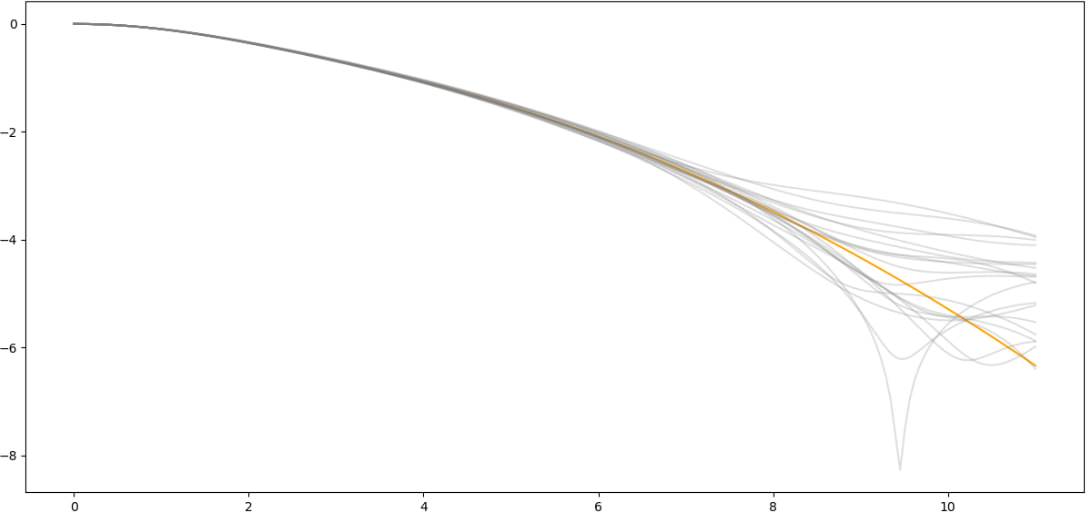}
    \caption{\label{fig1}The plot of $\Re(\log(\phi(u)))$ (orange line) and plots of $N=20$ realisations of its estimate $\Re(\log(\phi_n(u)))$ (grey lines).}
    \label{fig:enter-label}
\end{figure}

Figure~\ref{fig2} shows the quality of the estimates of $\sigma^2_1$ and $p$ obtained by our method and by the EM algorithm based on  \(n\) observations from  $N=50$ simulation runs. The number \(n\) takes the values $1000, 5000, 10000$. Both estimates converge to the true values of the parameters, but the accuracy of our estimate increases with the growth of  \(n\), compared to the estimates obtained by the EM algorithm.
\begin{figure}[h]
    \centering
    \begin{subfigure}[b]{0.45\textwidth}
        \centering
        \includegraphics[scale=0.22]{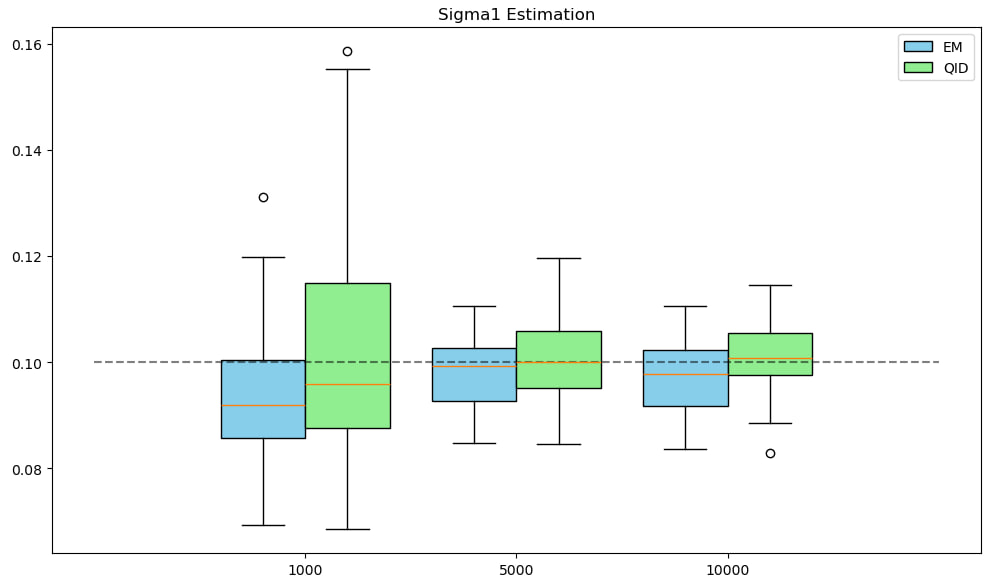}
        \caption{$\widehat{\sigma}_1^2$ }
        \label{fig:image1}
    \end{subfigure}
    \hfill
    \begin{subfigure}[b]{0.45\textwidth}
        \centering
        \includegraphics[scale=0.22]{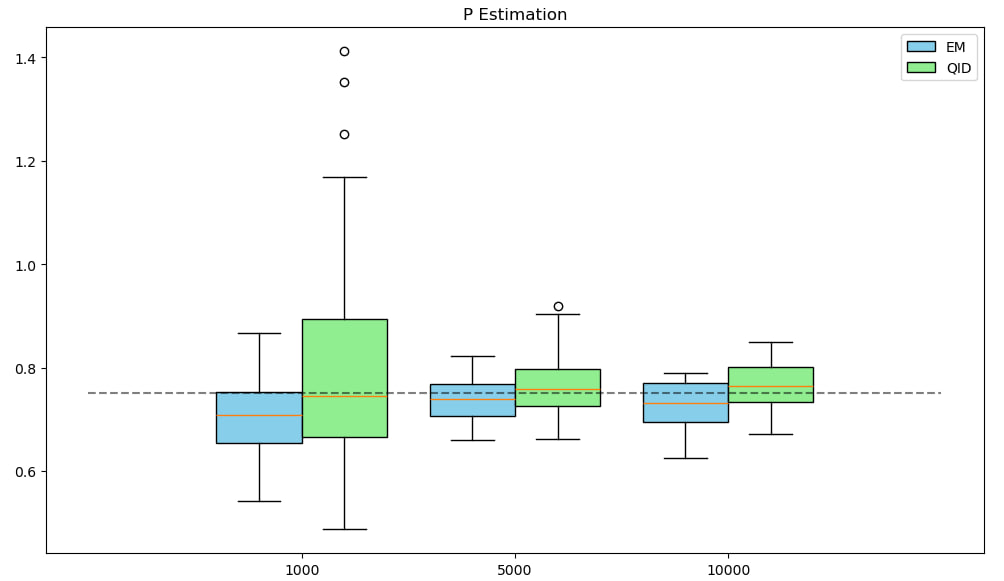}
        \caption{$\widehat{p}$}
        \label{fig:image2}
    \end{subfigure}
    
    \caption{\label{fig2}
    Boxplots for the estimates of \(\sigma_1^2\) and \(p\) obtained by the  EM algorithm (blue boxes) and our approach (green boxes) for the two-component normal mixture model.}
    \label{fig:allimages}
\end{figure}

Lastly, Figure~\ref{fig3} shows the estimates of $g^{\circ}(\cdot)$  obtained by both algorithms for different number of observations $n$. The quality of our algorithm improves as \(n\) increases, resulting in a performance comparable to that of the EM algorithm for large sample sizes.

\begin{figure}[h]
    \centering
    \begin{subfigure}[b]{0.45\textwidth}
        \centering
        \includegraphics[scale=0.2]{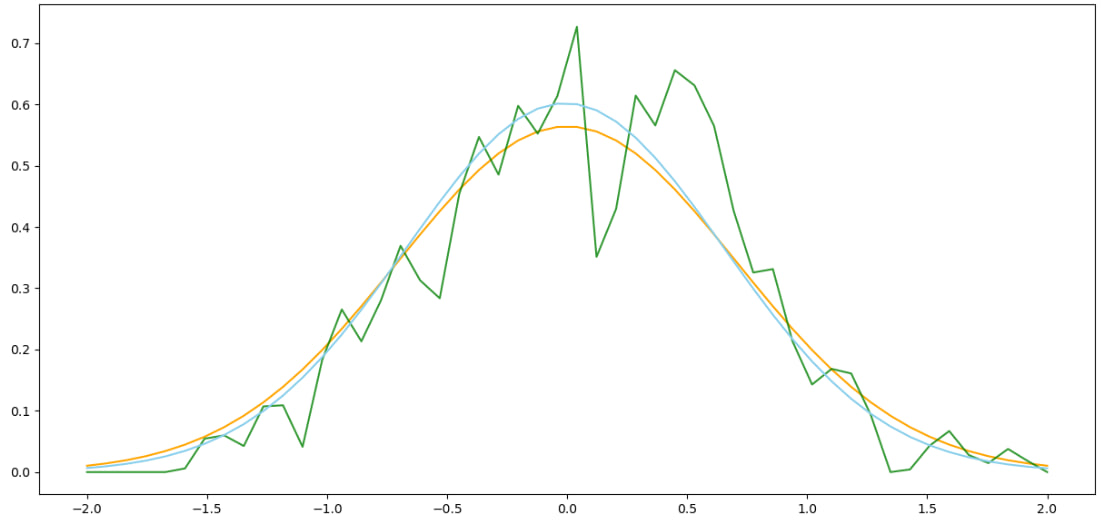}
        \caption{$n = 1000$}
        \label{fig:image1}
    \end{subfigure}
    \hfill
    \begin{subfigure}[b]{0.45\textwidth}
        \centering
        \includegraphics[scale=0.2]{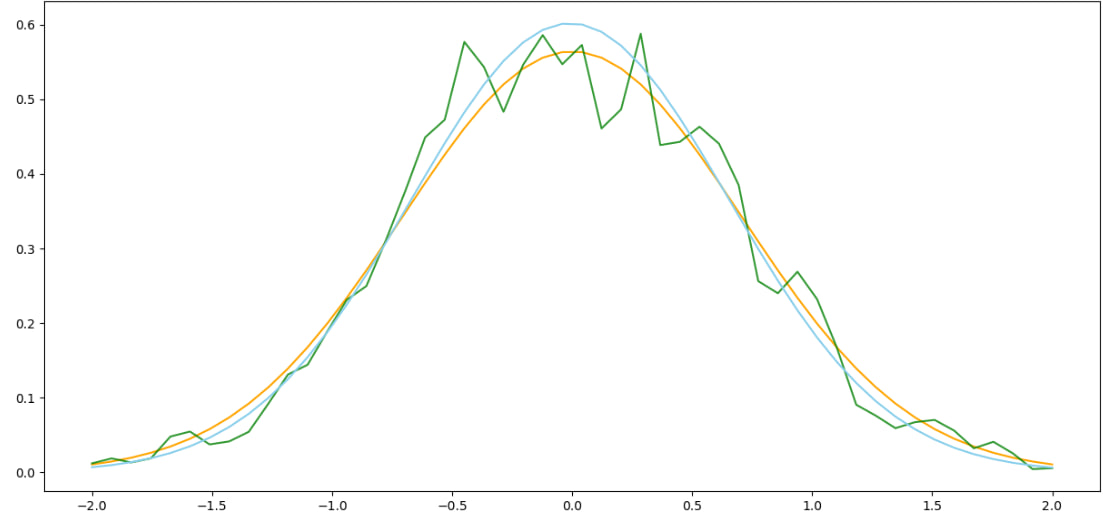}
        \caption{$n = 5000$}
        \label{fig:image2}
    \end{subfigure}
    \begin{subfigure}[b]{0.45\textwidth}
        \centering
        \includegraphics[scale=0.2]{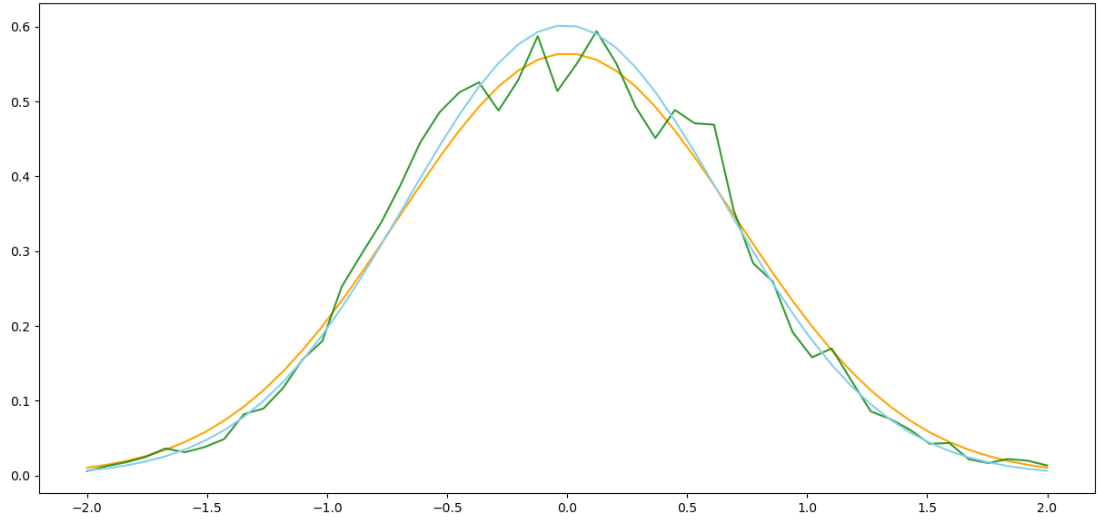}
        \caption{$n = 10000$}
        \label{fig:image3}
    \end{subfigure}
    
    \caption{\label{fig3} Plots of the true function $g^{\circ}(x)$ (orange line) and  its estimates, obtained by the  EM algorithm (blue line) and our inference method (green line) for the two-component normal mixture model.}
    \label{fig:allimages}
\end{figure}

\subsection{Bart Simpson distribution}
In~Example~6.1 from \cite{Wasserman}, the Bart Simpson distribution is defined via the density function
\begin{equation*}
    g(x) = p\varphi_{(0, \sigma_1)}(x) + \frac{1-p}{5}\sum_{j = 0}^{4} \varphi_{((j/2) - 1, \, \sigma_2)}(x),
\end{equation*}
where $p = 0.5, \sigma_1=0.1, \sigma_2=0.01$, and $\varphi_{(m, \sigma)}$ is a density of the normal distribution with parameters $m$ and $\sigma^2$. The name of this distribution is derived from the shape of its probability density function, which resembles the hair on Bart Simpson's head.

With this choice of parameters \(p, \sigma_1, \sigma_2\),  the conditions of Corollary~\ref{cor2} are not satisfied. Nevertheless, if  \(p=0.5+\delta\) with any (small) \(\delta\) and \(\sigma_1<\sigma_2\), then the conditions hold and the distribution is QID.

For this simulation study we take  \(\delta=0.001,  \sigma_1 = 0.05, \sigma_2 = 0.1.\) Figure~\ref{fig4} depicts the histogram of the simulated data for $n=1000$  observations. Note that there are 5 distinct ``spikes'' in the density function, although they are not as pronounced as in the original example.  Below we call this distribution  ``the modified Bart Simpson model''.

\begin{figure}[ht]
    \centering
    \includegraphics[scale=0.25]{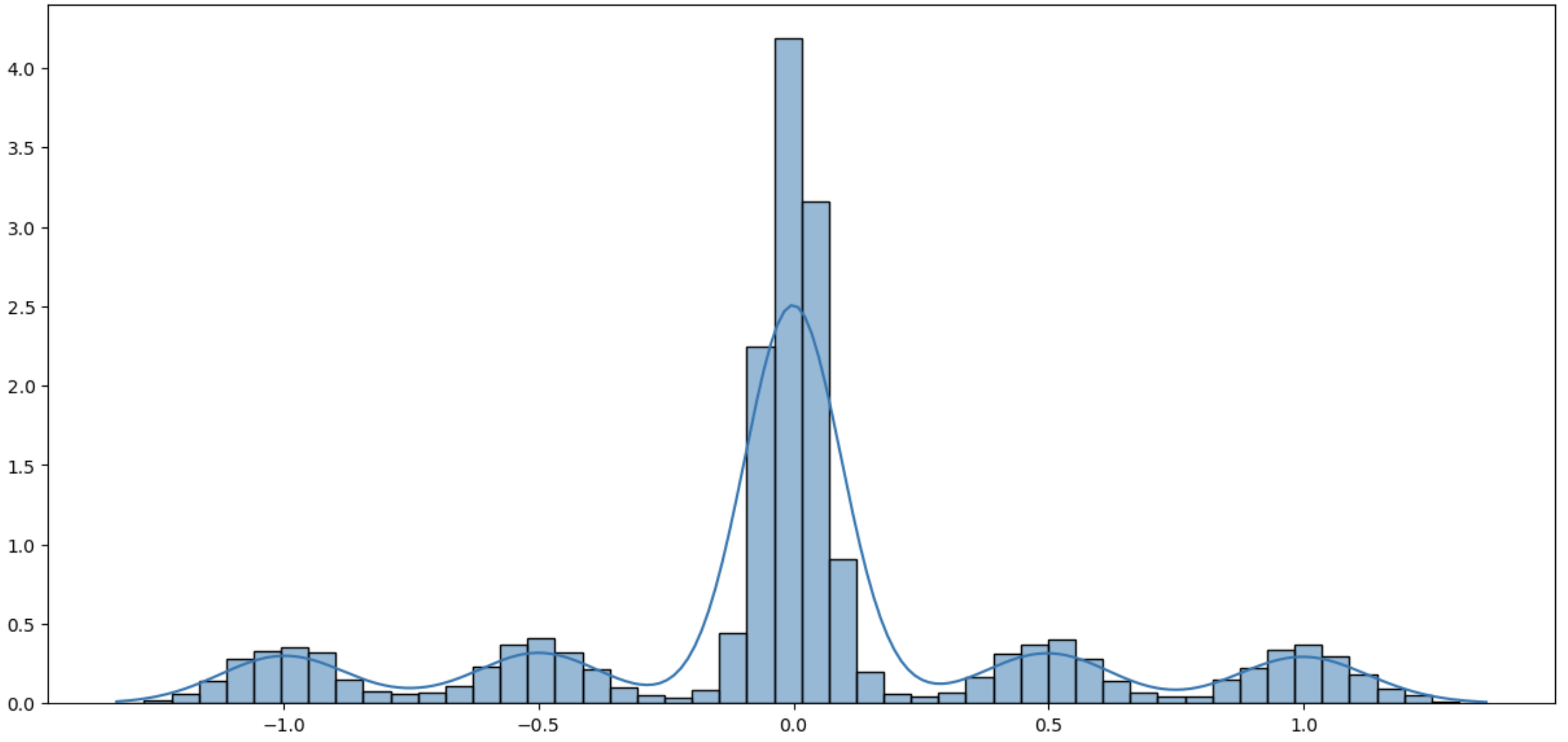}
    \caption{\label{fig4} Histogram from the modified Bart Simpson model and the graph of the true density $g(x)$.}
    \label{fig:enter-label}
\end{figure}

Figure~\ref{fig5} shows the comparison of the estimates of $\sigma_1^2$ and $p^*$  between the EM algorithm and our method  based on  \(n\) observations from  $N=50$ simulation runs.
In contrast to the previous example, the estimate obtained by the EM algorithm has a noticeably smaller deviation and converges more quickly. However,  our inference approach also leads to a good quality of estimation.

\begin{figure}[ht]
    \centering
    \begin{subfigure}[b]{0.45\textwidth}
        \centering
        \includegraphics[scale=0.18]{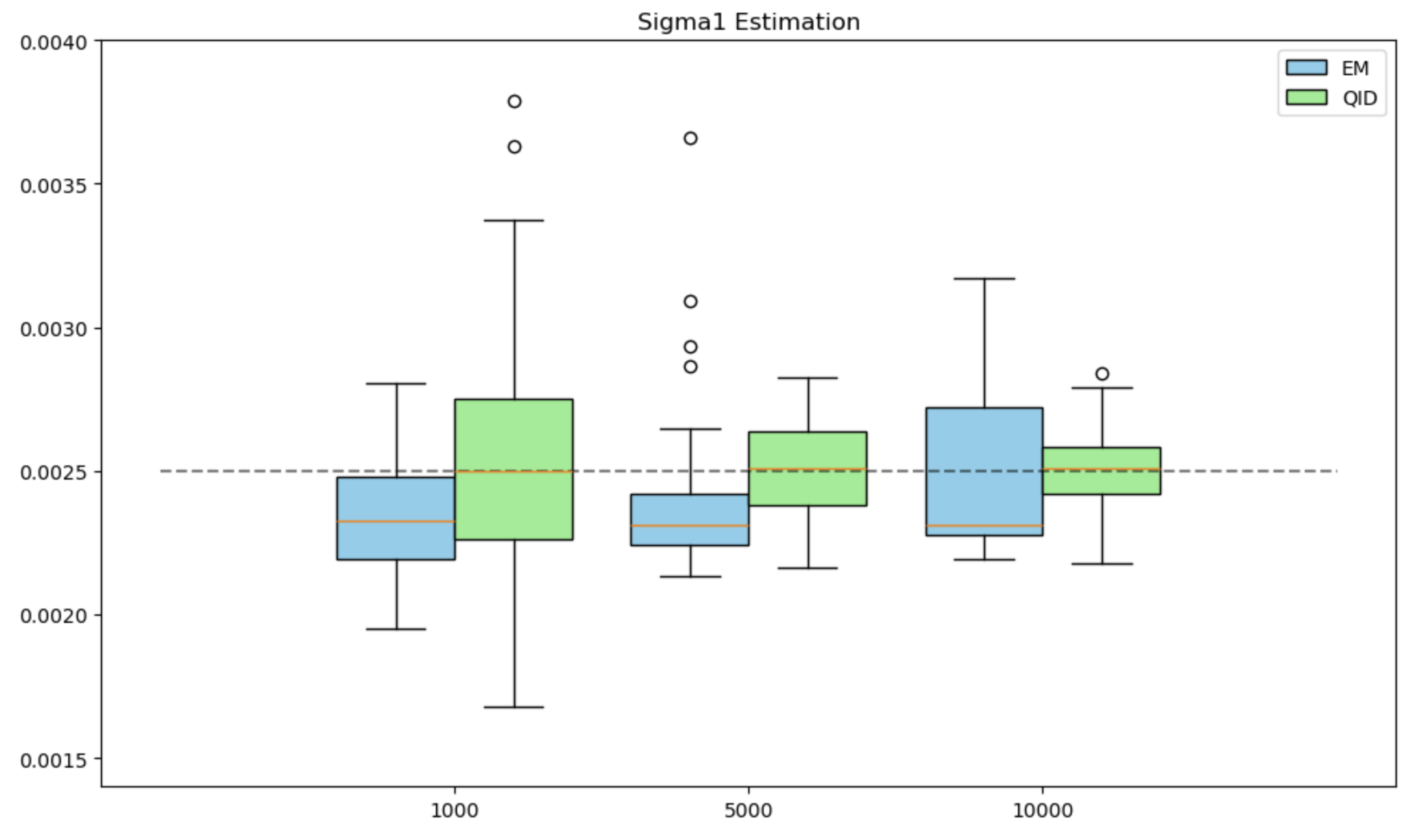}
        \caption{$\widehat\sigma^2_{1}$ }
        \label{fig:image1}
    \end{subfigure}
    \hfill
    \begin{subfigure}[b]{0.45\textwidth}
        \centering
        \includegraphics[scale=0.18]{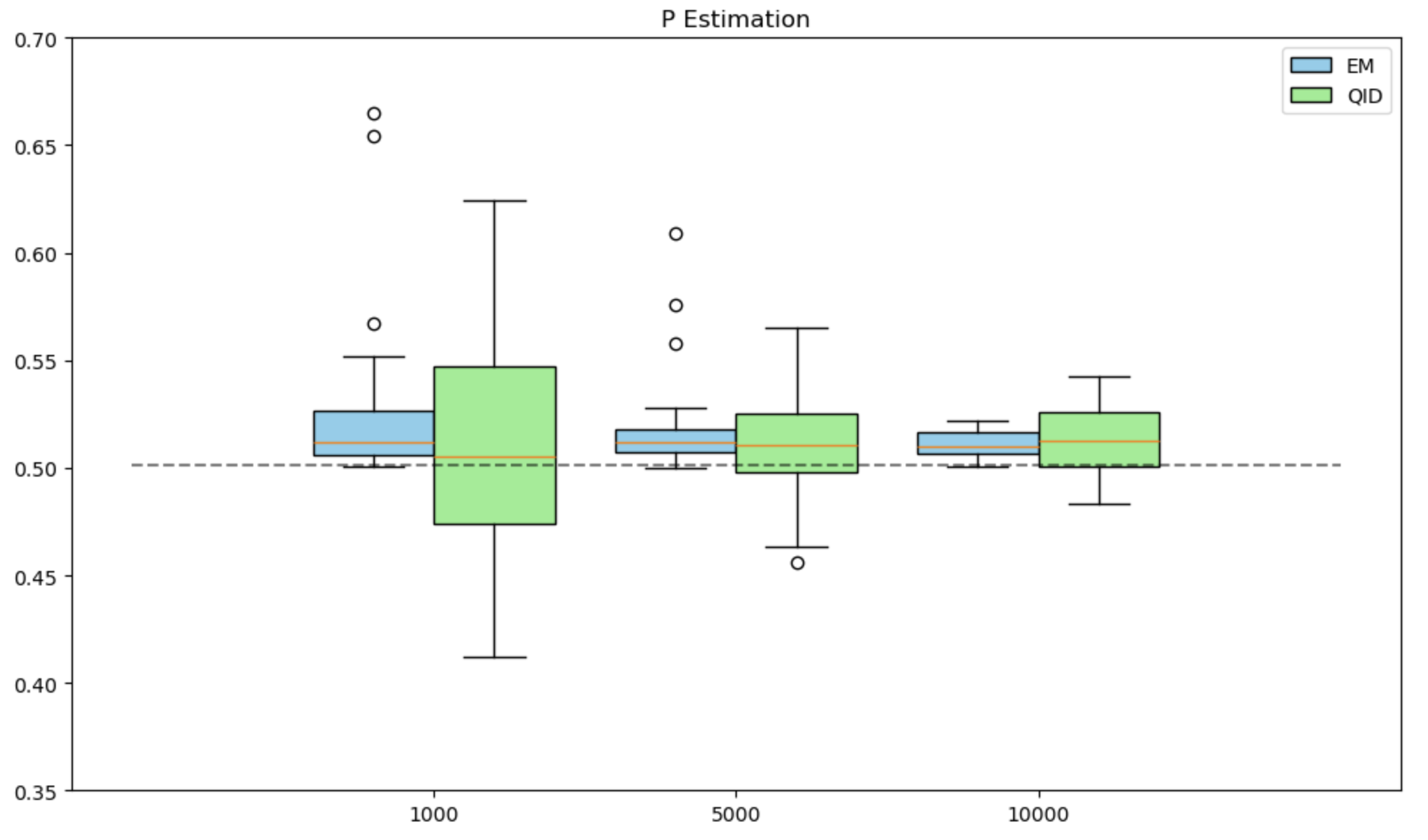}
        \caption{$\widehat{p^*}$}
        \label{fig:image2}
    \end{subfigure}
    
    \caption{\label{fig5}Boxplots for the estimates of \(\sigma_1^2\) and \(p\) obtained by the  EM algorithm (blue boxes) and our approach (green boxes) for the modified Bart Simpson model.}
    \label{fig:allimages}
\end{figure}

Figure~\ref{fig6} presents comparison between true mixture density $g(x)$ and its estimates \(g_n(x) \) obtained by our approach and the EM algorithm. We see that both estimates replicate all 5 spikes of the true density and have similar estimation quality at each of the density spikes.

\begin{figure}[ht]
    \centering
    \includegraphics[scale=0.35]{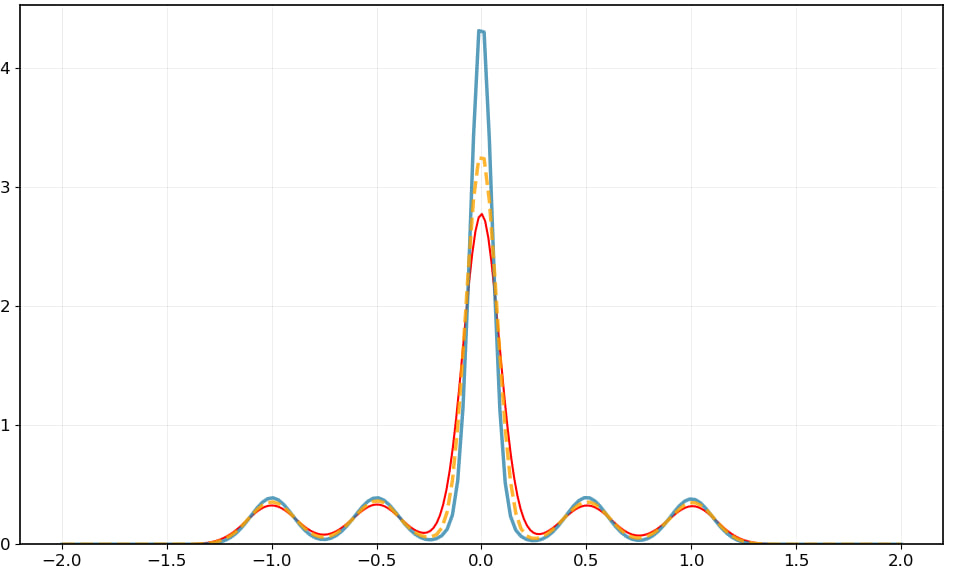}
    \caption{\label{fig6}
    Plots of the true function $g(x)$ (red line) and  its estimates, obtained by the  EM algorithm (blue line) and our inference method (orange line) for the modified Bart Simpson model.}
    \label{fig:enter-label}
\end{figure}


    

\subsection{Another example of model~\eqref{model1}}
In Example \ref{ex3}, we have considered the distribution 
\begin{equation*}
    \mu = p\tilde\mu_{\sigma_1} + (1-p) \mu^\circ,
\end{equation*}
where  \(\mu^\circ\) is ID with L{\'e}vy triplet \((\gamma, \bar\sigma^2, \nu)\). Recall that this distribution is QID, if \(p>1/2,\) and \(\bar{\sigma} > \sigma_1\).  For this numerical example, we fix $\mu^\circ$ as the convolution of standard Student's distribution with $d$ degrees of freedom and a normal distribution  with zero mean and  variance equal to \(\sigma_2^2\), where $0 < \sigma_1 < \sigma_2$. In this study we fix $d=3, p = 0.75$, $\sigma^2_1 = 0.2$, $\sigma^2_2 = 0.5$.

Figure~\ref{fig7} depicts the estimates of $\sigma^2_1$ and $p$  based on  \(n\) observations from  $N=50$ simulation runs. The number \(n\) takes the values
$1000, 5000, 10000$. Here we  again conclude that the estimates converge to the true values and the variances of the estimates decrease with the growth of \(n\).

\begin{figure}[ht]
    \centering
    \begin{subfigure}[b]{0.45\textwidth}
        \centering
        \includegraphics[scale=0.18]{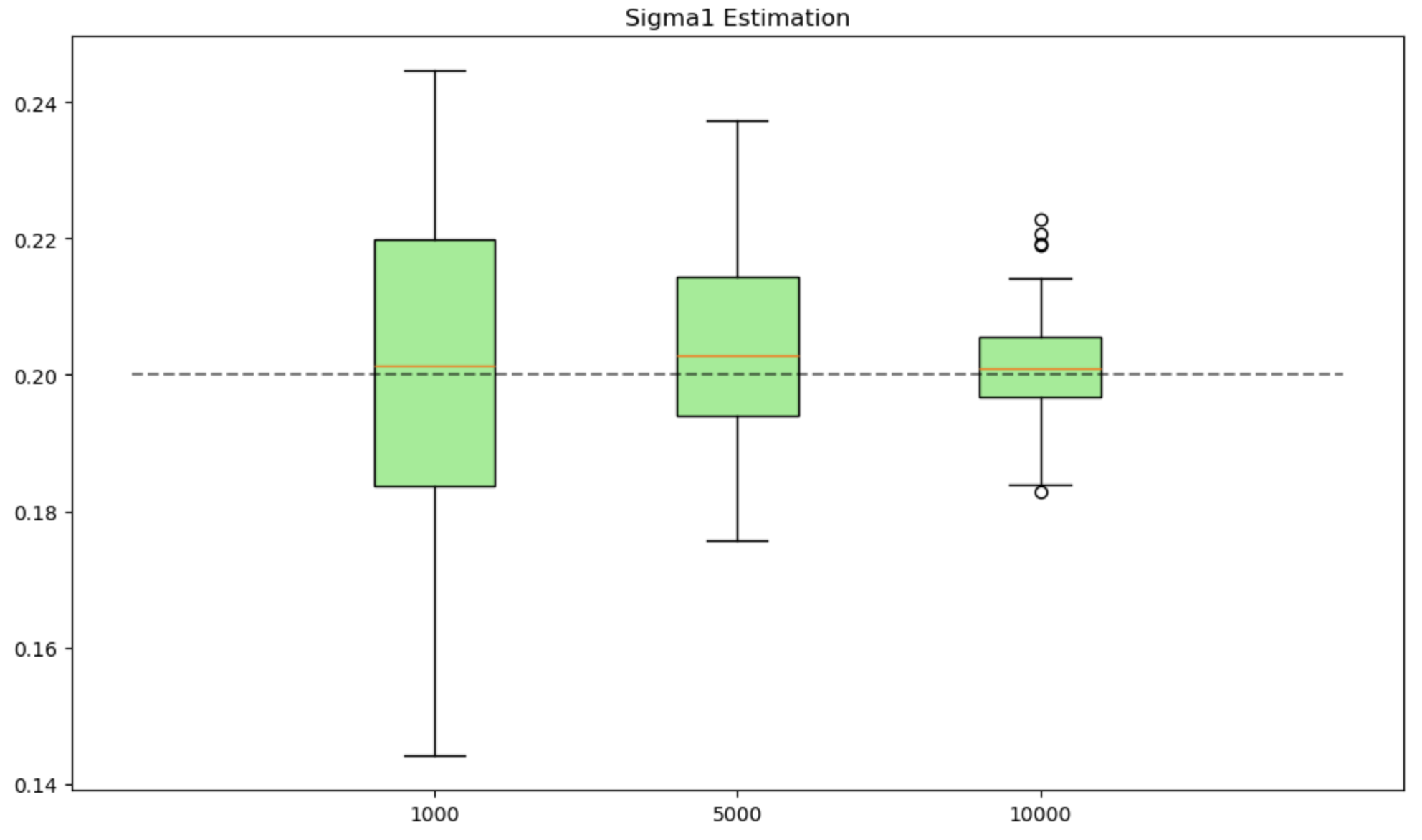}
        \caption{$\widehat\sigma^2_{1}$}
        \label{fig:image1}
    \end{subfigure}
    \hfill
    \begin{subfigure}[b]{0.45\textwidth}
        \centering
        \includegraphics[scale=0.18]{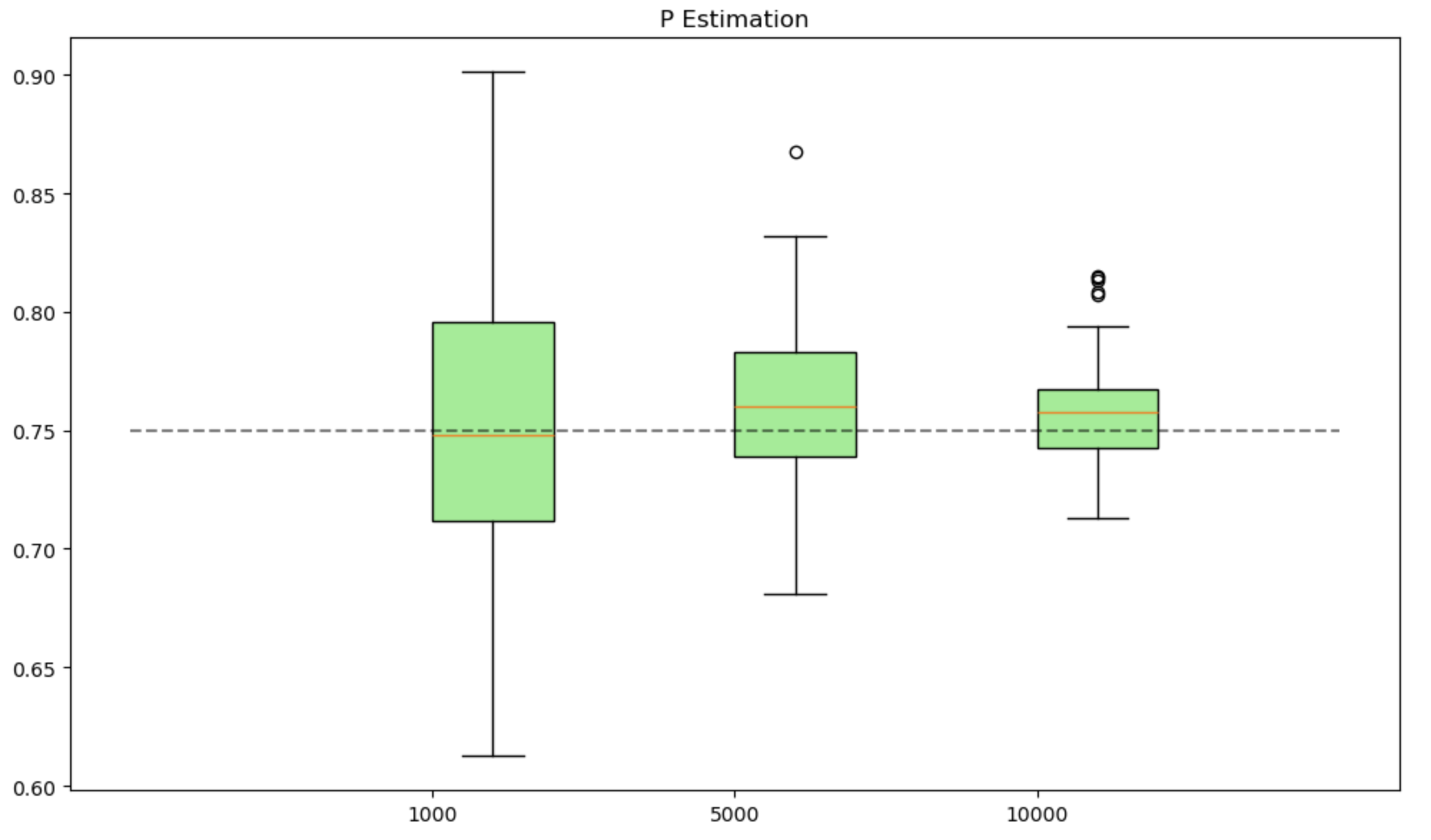}
        \caption{$\widehat{p}$}
        \label{fig:image2}
    \end{subfigure}
    
    \caption{\label{fig7} Boxplots for the estimates of \(\sigma_1^2\) and \(p\) for the case when the second component of the mixture is the convolution of Student's and normal distributions.  }
    \label{fig:allimages}
\end{figure}

Analogously to Figure~\ref{fig6}, Figure~\ref{fig9} depicts the decomposition of the density \(g\) into two  parts. It can be seen that the first part, which corresponds to the normal component  (green line), is the main part of the resulting mixture density, while the second part (orange line) makes the quality of the estimate better.
\begin{figure}[ht]
    \centering
    \includegraphics[scale=0.25]{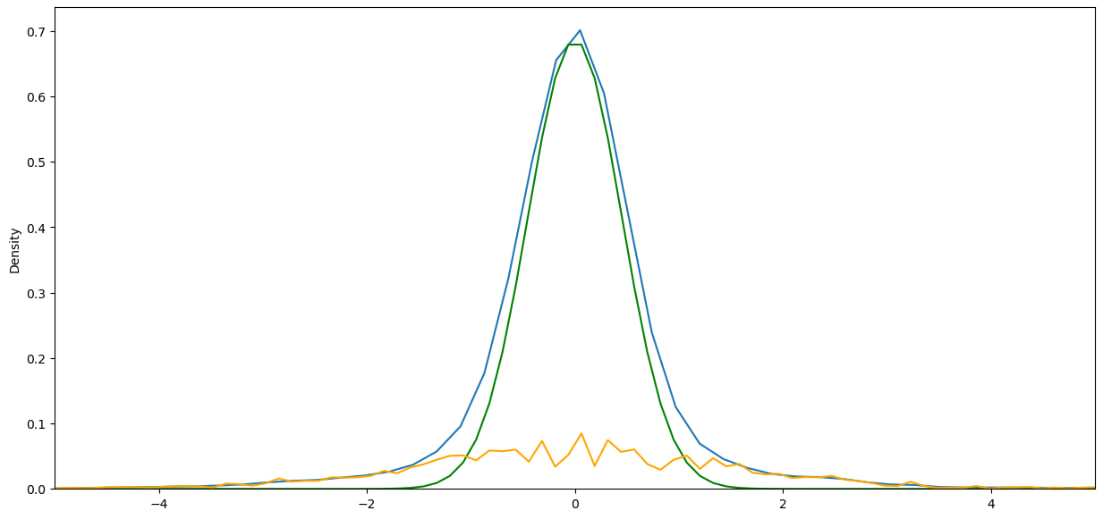}
    \caption{\label{fig9} Plot of  the true density $g(x)$  (blue line) and plots of the first and the second components (green and orange lines respectively) for the case when the second component of the mixture is the convolution of Student's and normal distributions.}
    \label{fig:enter-label}
\end{figure}

Lastly, we consider more precisely the estimation of the density \(g^{\circ}\) of the second component. Figure~\ref{fig10} illustrates the accuracy  of our estimate $g^{\circ}_n(x)$ for different \(n\), which coincides with $g^{\circ+}_{n}(x)$ in this case. The quality of this estimate increases with the growth of $n$.

\begin{figure}[h]
    \centering
    \begin{subfigure}[b]{0.45\textwidth}
        \centering
        \includegraphics[scale=0.17]{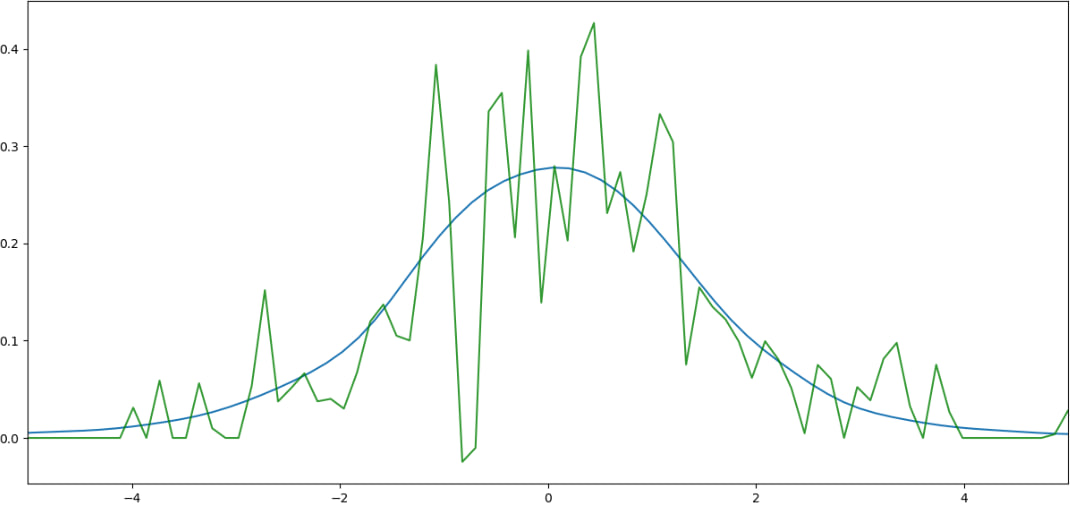}
        \caption{$n = 1000$}
        \label{fig:image1}
    \end{subfigure}
    \hfill
    \begin{subfigure}[b]{0.45\textwidth}
        \centering
        \includegraphics[scale=0.17]{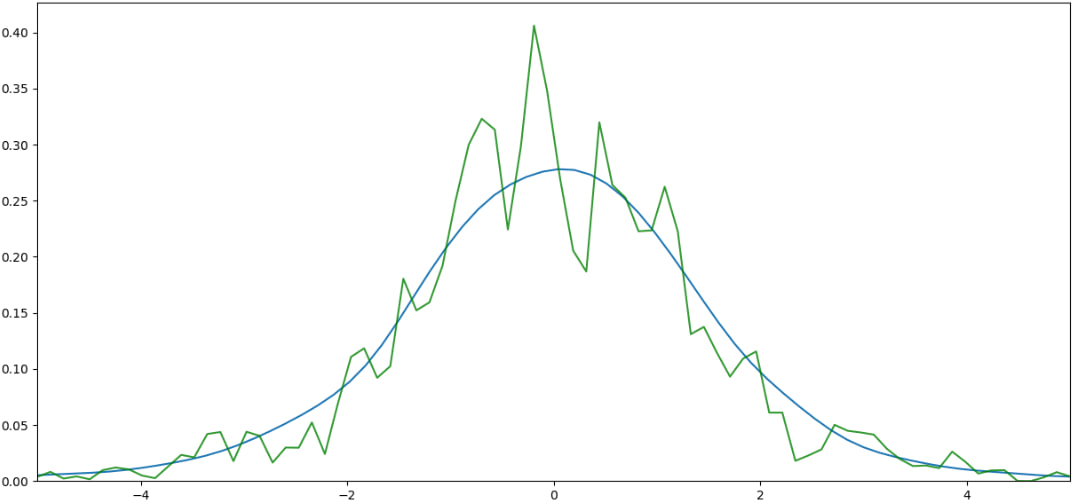}
        \caption{$n = 5000$}
        \label{fig:image2}
    \end{subfigure}
    \begin{subfigure}[b]{0.45\textwidth}
        \centering
        \includegraphics[scale=0.17]{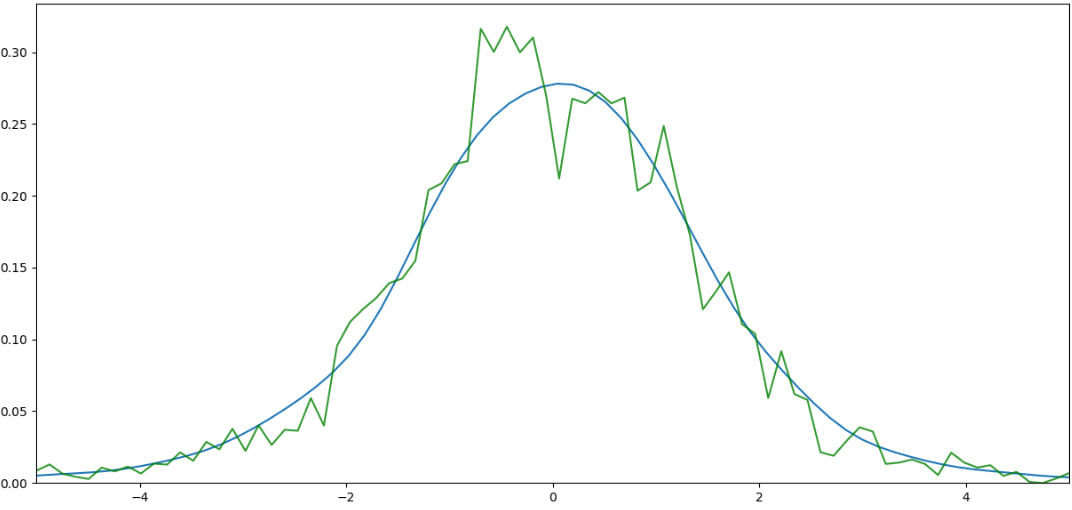}
        \caption{$n = 10000$}
        \label{fig:image3}
    \end{subfigure}
    
    \caption{\label{fig10} Plots of the true functions $g^{\circ}(x)$ (blue line) and their estimates $g_n^{\circ}(x)$ (green line) based on \(n\) observations for the case when the second component of the mixture is the convolution of Student's and normal distributions.}
    \label{fig:allimages}
\end{figure}

\section{Proofs} \label{sec7}
\subsection{Proof of Theorem~\ref{thm1}}
\textbf{1.} Below we will prove the upper bound for the mean-squared error of \(\sigma_n^2\). The proofs for \(\lambda_n^*\) and \(\gamma_n^*\) follow the same lines. 


Using the representation~\eqref{sigma_sol} and the properties~\eqref{prop1}, we get 
\begin{eqnarray*}
    \sigma^2_n - \sigma^2 &=& \int_{\R}w^{U_n}_{\sigma^2}(u)\Re\bigl(\log(\phi_n(u)) - \log(\phi(u))\bigr)du + \sigma^2\int_{\R}\frac{u^2}{2} w^{U_n}_{\sigma^2}(u) du \\
    &&\hspace{2cm} + \int_{\R}w^{U_n}_{\sigma^2}(u)\Re\bigl(\log(\phi(u))\bigr)du + \lambda^* \int_{\R}w^{U_n}_{\sigma^2}(u)du \\
    &=& \int_{\R} w^{U_n}_{\sigma^2}(u)\Re\bigl(\log(\phi_n(u)) - \log(\phi(u))\bigr)du \\
&&    \hspace{6cm}+ \int_{\R}w^{U_n}_{\sigma^2}(u)
\Re\bigl(\F [s](u)\bigr)du,
\end{eqnarray*}
where the second equality follows from~\eqref{phiexp}.
Since \(|\log(1+z) -z| \leq 2 |z|^2\) for any \(|z| <1/2\), we have on the event \(\A_n\) for $n$ large enough  \begin{eqnarray*}
    \log(\phi_n(u)) - \log(\phi(u)) =
    \log\Bigl(\frac{\phi_n(u) - \phi(u)}{\phi(u)} +1\Bigr) 
    \label{log_eq} =  \frac{\phi_n(u) - \phi(u)}{\phi(u)} + R_n (u),
\end{eqnarray*}
where \( R_n (u) =  O\Bigl(|(\phi_n(u) - \phi(u))/\phi(u)|^2\Bigr)\), \(u \in \R\). Therefore, on $\A_n$ we get 
\begin{eqnarray*}
    \Bigl|\sigma^2_n - \sigma^2
\Bigr|
\leq   I_1 + I_2 + I_3,
\end{eqnarray*}
where 
\begin{eqnarray*}
I_1 &:=& \Bigl|\int_{\R} w^{U_n}_{\sigma^2}(u)\Re\Bigl(\frac{\phi_n(u) - \phi(u)}{\phi(u)}\Bigr)du\Bigr|,\\
I_2 &:=& \Bigl|\int_{\R} w^{U_n}_{\sigma^2}(u) \Re(\mathrm{R}_n (u)) du \Bigr|,\\
I_3 &:=& \Bigl| \int_{\R}w^{U_n}_{\sigma^2}(u)\Re\bigl(\F [s ](u)\bigr)du\Bigr|.
\end{eqnarray*}Below we separately consider these three terms. 
For the terms \(I_1\) and \(I_2\) we have 
\begin{align*}
    I_j &\leq \max_{u \in [-U_n,U_n]}\Bigl| \frac{\phi_n(u) - \phi(u)}{\phi(u)}\Bigr|^{j} \Bigl(\int_{\R} |w^{U_n}_{\sigma^2}(u)| du \Bigr) 
     \lesssim \frac{\chi_n^{j}}{U_n^{2}}, \qquad j=1,2.
\end{align*}
Therefore,  \(I_2 \lesssim I_1\), and 
due to Lemma~\ref{lem1} (ii), 
\begin{eqnarray*}
I_1  \lesssim  \exp(C)
\frac{\sqrt{\log(nU_n^2)}}{\sqrt{n}U_n^{2}}
\exp(\frac{\overline{\sigma}^2U_n^2}{2}).
\end{eqnarray*}
For $I_3,$ we apply the Plancherel identity 
\begin{align*}
    I_3 &\leq |\int_{\R}w^{U_n}_{\sigma^2}(u)\F \bigl[s \bigr](u)du| \\
    &= 2\pi |\int_{\R_+} s^{(r)}(x)\overline{\F^{-1}} \bigl[ \frac{w^{U_n}_{\sigma^2}(t)}{(it)^r}\bigr](x)dx|  \\
    &\leq U_n^{-(r + 3)}||s^{(r)}||_{\infty} \cdot ||\F \bigl[ \frac{w^{1}_{\sigma^2}(u)}{u^r}\bigr]||_{L^1},
\end{align*}
yielding the upper bound \(I_3 \lesssim C  U_n^{-(r + 3)}\). Combining these bounds, we arrive at the desired result.\newline

\textbf{2.} Now we consider the upper bound for the integrated squared error of \(s_n.\) Denote the estimate of the Fourier transform of \(s\) by \[\widehat{\F\bigl[s \bigr]}(u)=
\log\bigl(
\phi_n(u) 
\bigr) - \i \gamma_n^* u + \frac{1}{2} \sigma_n^2 u^2+\lambda^*_n.\]
We have on the event \(\A_n\),
\begin{eqnarray*}
\bigl|
\widehat{\F\bigl[s \bigr]}(u)  - 
\F\bigl[s \bigr](u)
\bigr| 
\lesssim \chi_n + \frac{\mathcal{R}_n}{U_n} u +  \frac{\mathcal{R}_n}{U_n^2} u^2 +\mathcal{R}_n.
\end{eqnarray*}
Next, using the Parseval identity we get 
\begin{eqnarray*} 
\int_{\R} \bigl( 
s_n(x) - s(x) 
\bigr)^2 dx &=& 
\int_{\R} \bigl( 
\F^{-1}\bigl[\widehat{\F\bigl[s \bigr]}(\cdot)w_s(\cdot / T_n)\bigr](x) - s(x) 
\bigr)^2 dx \\ 
&=&
\frac{1}{2\pi}
\int_{\R} \bigl|
\widehat{\F\bigl[s \bigr]}(u)w_s(u / T_n) - \F[s](u) 
\bigr|^2 du \\
&\lesssim& 
\int_{\R} \bigl|
\widehat{\F\bigl[s \bigr]}(u) - \F[s](u) 
\bigr|^2 \bigl( w_s(u / T_n) \bigr)^2du 
\\
&& + \int_{\R} \bigl|
\F[s](u) 
\bigr|^2 \bigl| w_s(u / T_n) -1 \bigr|^2du 
=: J_1 +J_2.
\end{eqnarray*} 
Now let us consider two terms separately.  On the event \(\A_n,\)
\begin{eqnarray*}
J_1 &\lesssim& \bigl(\chi_n^2 + \mathcal{R}_n^2 \bigr) T_n +  \frac{\mathcal{R}_n^2}{U_n^2} T_n^2 + \frac{\mathcal{R}_n^2}{U_n^4} T_n^3,
\end{eqnarray*}
because for any natural \(k\), \(\int_{-1}^{1} v^k (w_s(v))^2 dv \leq \int_{-1}^{1} (w_s(v))^2 dv\) is bounded. As for the second summand, we get 
\begin{eqnarray*}
J_2 &\lesssim&  T_n^{-2r} \int_{\R} \bigl|
\F[s](u) 
\bigr|^2 u^{2r} du =
T_n^{-2r} \int_{\R} \bigl|
\F[s^{(r)}](u) 
\bigr|^2 du \\
&\lesssim& T_n^{-2r} \|
s^{(r)}
\|^2 \lesssim C T_n^{-2r}.
\end{eqnarray*}
This observation completes the proof.
\subsection{Proof of Theorem~\ref{thm2}}
The proof is similar to the proof of Theorem~\ref{thm1} with the different approach for the estimation of the term $I_3$. Using the properties of the L{\'e}vy measure for the QID distributions from the class $\mathcal{S}^*(r, \bar{\sigma}, C,\underline{p})$, we get the following upper bound
\begin{align*}
    I_3 &=  \Bigl|\int_{\R_+} w^{U_n}_{\sigma^2}(u) \F \Bigl[ \sum_{m = 1}^{\infty} \frac{(-1)^{m + 1}}{m} (q/p)^m\Lambda^{*m} \Bigr](u)du \Bigr| \\
    &= \Bigl|\int_{\R_+} w^{U_n}_{\sigma^2}(u) \sum_{m = 1}^{\infty} \frac{(-1)^{m + 1}}{m} (q/p)^m \bigl(\H(u)\bigr)^m du \Bigr| \\
    &= \Bigl|U_n^{-2}\sum_{m=1}^{\infty} \frac{(-1)^{m + 1}}{m} (q/p)^m \int_{\varepsilon}^{1} w^{1}_{\sigma^2}(u)\bigl(\H(uU_n)\Bigr)^m du \bigr| \\
    &\leq U_n^{-2} \sum_{m=1}^{\infty} \frac{(q/p)^m}{m}   \int_{\varepsilon}^{1} |w^{1}_{\sigma^2}(u)\bigl|\cdot \bigl| \H(uU_n)\bigr|^m du.
\end{align*}
Due to our assumption~\eqref{ss}, we get
for  large \(n\) \begin{align*}
    I_3
    &\leq c_1 U_n^{-2} \int_{\varepsilon}^{1} |w^{1}_{\sigma^2}(u)| e^{-c_2(u U_n)^{\gamma}} du  \cdot \sum_{m=1}^{\infty} \frac{(q/p)^m}{m} \\
    &\leq c_1U_n^{-2} e^{-c_2 \eps^{\gamma} U_n^{\gamma}}\bigl(\int_{\varepsilon}^{1} |w^{1}_{\sigma^2}(u)| du \bigr)  \cdot \log\Bigl( \frac{p}{p-q} \Bigr)   \\
    &\lesssim c_1  \log\Bigl( \frac{1}{2 \underline{p}-1} \Bigr)  U_n^{-2} \e^{-c_2 \eps^{\gamma} U_n^{\gamma}}.
\end{align*}
Combining this inequality with the upper bounds for \(I_1\) and \(I_2\), which were obtained in the previous section, completes the proof.

\subsection{Proof of Remark~\ref{remark2}}
To show the upper bound for the estimate $p_n= e^{-\lambda^*_n}$, we use the Taylor series decomposition with the stochastic remainder term. For a \(k\)-times differentiable function \(f: \C \to \C,\) this decomposition reads as 
\begin{multline}
    f(z) = \sum_{j=0}^{k-1}\frac{f^{(j)}(a)}{j!}(z - a)^j \\+ \frac{1}{(k - 1)!}\E_{\tau}\Bigl[
(1 - \tau)^{k - 1} f^{(k)}(a + \tau(z - a))\Bigr]
(z - a)^k, 
\label{taylor}
\end{multline}
where $\tau$ is uniformly distributed on $[0, 1]$, see \cite{panov_morozova}, Appendix~A.1. Therefore, \begin{align*}
    p_n - p &= e^{-\lambda^*_n} - e^{-\lambda^*} \\ 
    &= -e^{-\lambda^*}\Bigl((\lambda^*_n - \lambda^*) - (\lambda^*_n - \lambda^*)^2 \E_{\tau} \Bigl[ (1 - \tau)e^{-\tau(\lambda^*_n - \lambda^*)}\Bigr]\Bigr),
\end{align*}
and we conclude that 
\begin{align}\label{p_estim}
    \bigl|p_n - p\bigr| &\leq |\lambda^*_n - \lambda^*| + |\lambda^*_n - \lambda^*|^2e^{|\lambda^*_n - \lambda^*|} \lesssim \breve{\mathcal{R}}_n,
\end{align}
and the result follows.
\subsection{Proof of Theorem~\ref{thm3}}
First note that
\begin{equation*}
    \E\bigl(g^{\circ+}_{n}(x) - g^{\circ}(x)\bigr)^2 \leq \E\bigl(g^{\circ}_{n}(x) - g^{\circ}(x)\bigr)^2, \quad \forall x\in \R,
\end{equation*}
since $g^{\circ}(x) \geq 0$, $\forall x\in \R$. Therefore, it is sufficient to establish  the result of Theorem~\ref{thm3} for $g^{\circ}_{n}(x)$ instead of $g^{\circ+}_{n}(x)$. We have 
\begin{align*}
    \int_{\R} \E &\Bigl[ \bigl(g^{\circ}_{n} (x) - g^{\circ}(x)\bigr)^2 | \A_n\Bigr] dx \\
    &= \int_{\R} \E \Bigl[ \Bigl(\frac{g_n(x) - p_n\varphi_{\sigma_n}(x)}{1-p_n}- \frac{g(x) - p\varphi_{\sigma}(x)}{1-p}\Bigr)^2 | \A_n\Bigr]dx\\
    &= \int_{\R} \E \Bigl[ \frac{\Bigl((1 - p)(I_1 +I_2) +I_3 + I_4\Bigr)^2}{(1-p_n)^2(1-p)^2}| \A_n\Bigr] dx \\\ 
    &\lesssim \frac{1}{\bigl(1- \overline{p}\bigr)^4}\int_{\R} \E \Bigl[ \bigl(I_1^2 +I_2^2 + I_3^2 +I_4^2\bigr) |\A_n\Bigr] dx,
\end{align*}
where
\begin{align*}
    I_1 &= |g_n(x) - g(x)|, \qquad 
    I_2 = p_n |\varphi_{\sigma_n}(x) - \varphi_{\sigma}(x)|, \\
    I_3 &= |p-p_n|\varphi_{\sigma}(x),  \qquad 
    I_4 = |p - p_n|g(x),
\end{align*}
and we use that $\underline{p} \leq p \leq \overline{p}$ and 
\begin{eqnarray*}
|1-p_n| \geq |1-p| - |p-p_n| \geq |1-\overline{p}| +o(1), \qquad n \to \infty,
\end{eqnarray*}
on the event \(\A_n\).
Below we consider the terms  \(I_1,...,I_4\) separately. For the first term, we use Theorem 1.3 from \cite{Tsybakov}, which yields
\begin{equation}\label{thm3_proof2}
    \E \bigl[ I_1^2 | \A_n \bigr] \lesssim \frac{1}{nh}\int_{\R} K^2(u)du + \frac{h^4}{4}(\int_{\R}|u^2K(u)|du)^2\int_{\R}(g''(x))^2dx.
\end{equation}
Note that here we use that  \(g'' \in L^2(\R)\) due to the condition \(g_\circ'' \in L^2(\R).\)
For the second term, we use the  decomposition~\eqref{taylor} with the function \(f(z) = \varphi_{z}(x)\) and any fixed \(x \in \R\). Note that  on the event \(\A_n,\) $$ \xi := \sigma^2 + \tau(\sigma^2_n - \sigma^2) \in \bigl[ \underline{\sigma}^2, \overline{\sigma}^2 \bigr] $$for \(n\) large enough, see Remark~\ref{remark2}. This observation yields the following upper bound for the difference between \(\varphi_{\sigma_n}\) and  \(\varphi_{\sigma}\):
\begin{align*}
    |\varphi_{\sigma_n}(x) - \varphi_{\sigma}(x) | &= \Bigl|\frac{1}{\sqrt{2\pi}\sigma_n}\exp(\frac{-x^2}{2\sigma^2_n}) - \frac{1}{\sqrt{2\pi}\sigma}\exp(\frac{-x^2}{2\sigma^2})\Bigr| \\
    &= \Bigl|\E_{\tau}\Bigl[\frac{1}{\xi}(\frac{x^2}{\xi} - 1)\varphi_{\xi}(x)\Bigr](\sigma^2_n - \sigma^2)\Bigr| \\
    &\lesssim \frac{1}{\underline{\sigma}^3}\bigl(\frac{x^2}{\underline{\sigma}^2} + 1\bigr) \exp(-\frac{x^2}{2\overline{\sigma}^2}) \frac{\breve{\mathcal{R}}_n}{U_n^2}.
\end{align*}
Therefore, 
\begin{align}\label{thm3_proof4}
    \E \bigl[ I^2_2 | \A_n \bigr] &\lesssim  \E \bigl[p_n^2 | A_n \bigr]\int_{\R} \Bigl( \frac{1}{\underline{\sigma}^3}(\frac{x^2}{\underline{\sigma}^2} + 1)\exp(-\frac{x^2}{2\overline{\sigma}^2})\Bigr)^2 dx \cdot \frac{\breve{\mathcal{R}}_n^2}{U_n^4} \nonumber \\
    &=   \E \bigl[p_n^2 | \A_n \bigr] \cdot \frac{\overline{\sigma}\sqrt{\pi}}{\underline{\sigma}^6}\Bigl(\frac{3}{4}\bigl(\frac{\overline{\sigma}}{\underline{\sigma}}\bigr)^4 +\bigl(\frac{\overline{\sigma}}{\underline{\sigma}}\bigr)^2 + 1\Bigr)  \cdot \frac{\breve{\mathcal{R}}_n^2}{U_n^4}  \nonumber \\
    &\lesssim  \frac{\overline{\sigma}}{\underline{\sigma}^6}\Bigl(\frac{3}{4}\bigl(\frac{\overline{\sigma}}{\underline{\sigma}}\bigr)^4 +\bigl(\frac{\overline{\sigma}}{\underline{\sigma}}\bigr)^2 + 1\Bigr) \cdot \frac{\breve{\mathcal{R}}_n^2}{U_n^4} ,
\end{align}
since \(\E \bigl[p_n^2 | \A_n \bigr]  \lesssim\E \bigl[(p_n - p )^2 | \A_n \bigr] + p^2 \leq 1\) for \(n\) large enough and 
\begin{equation*}
	\int_{\R} x^ne^{-x^2}dx = \frac{n!\sqrt{\pi}}{2^n(n/2)!}
\end{equation*} for even $n$.  Next, we get for the third term
\begin{align}\label{thm3_proof3}
    \E \bigl[ I_3^2 | \A_n \bigr] &= \E \bigl[ (p_n - p)^2 | \A_n \bigr] \cdot \int_{\R} \varphi^2_{\sigma}(x)dx
    \nonumber \\
    &= \frac{1}{2\sqrt{\pi}\sigma}\E \bigl[ (p_n - p)^2 | \A_n \bigr] \lesssim \frac{1}{\underline{\sigma}}\breve{\mathcal{R}}_n^2.\end{align}
Analogously, for the term $\E I^2_4$ we get
\begin{align}\label{thm3_proof4}
    \E \bigl[ I_4^2 | \A_n \bigr] &= \E \bigl[ (p_n - p)^2 | \A_n \bigr] \cdot \int_{\R} g^2(x)dx   \lesssim  \int_{\R} g^2(x) dx \cdot \breve{\mathcal{R}}_n^2.
\end{align}
Combining these bounds, we arrive at desired result.

\section*{Acknowledgment}
The article was prepared within the framework of the HSE University Basic Research Program.
\section*{Data Availability Statement}
The data that support the findings of this study are available from the corresponding author, Vladimir Panov, upon reasonable request.

\bibliographystyle{apalike}

\end{document}